\documentclass[onefignum,onetabnum]{siamonline171218}

\usepackage{lipsum}
\usepackage{amsfonts}
\usepackage{graphicx}
\usepackage{epstopdf}
\usepackage{xr}
\usepackage{algorithm}
\usepackage[noend]{algpseudocode}
\usepackage{centernot}
\usepackage{cite}
\usepackage{subfigure}
\usepackage{amssymb,amsmath}
\usepackage{caption}
\usepackage{arydshln}
\usepackage{comment}
\ifpdf
  \DeclareGraphicsExtensions{.eps,.pdf,.png,.jpg}
\else
  \DeclareGraphicsExtensions{.eps}
\fi

\DeclareMathOperator*{\argmin}{arg\,min}

\usepackage{enumitem}
\setlist[enumerate]{leftmargin=.5in}
\setlist[itemize]{leftmargin=.5in}

\newcommand{\norm}[1]{\left\lVert#1\right\rVert}

\newsiamremark{remark}{Remark}
\newsiamremark{hypothesis}{Hypothesis}
\crefname{hypothesis}{Hypothesis}{Hypotheses}
\newsiamthm{claim}{Claim}
\newtheorem{example}{Example}

\def\fskip#1{}

\def\rhoapc{\rho_{\text{APC}}}
\def\rhodhbm{\rho_{\text{HBM}}}

\def\R{\mathbb{R}}

\def\tilproj{\tilde P}

\def\phimin{\phi_{\text{min}}}

\def\thetamin{\theta_{\min}}

\def\row{\mathcal R}

\def\thetah{\theta_{\text{H} } }

\def\rhodgd{\rho_{\text{DGD} }}

\headers{Distributed Linear Solvers and Data Heterogeneity}{Boris Velasevic, Rohit Parasnis, Christopher G.\ Brinton, Navid Azizan}

\title{Distributed Linear Solvers and Data Heterogeneity\thanks{This paper was presented in part at the IEEE Conference on Decision and Control, Singapore, December 2023 \cite{velasevic2023effects_CDC}.\funding{This work was supported in part by MathWorks, by the MIT-IBM Lab, by the MIT-Google Program for Computing Innovation, by the MIT-Amazon Science Hub, and by NSF through grant CNS-2146171.}}}

\author{Boris Velasevic*\thanks{Massachusetts Institute of Technology, Cambridge, MA (\email{bvelasev@caltech.edu}), * denotes equal contribution, work done while the author was at MIT} 
\and  Rohit Parasnis*\thanks{Massachusetts Institute of Technology, Cambridge, MA (\email{rohit100@mit.edu}), * denotes equal contribution} 
\and Christopher G.\ Brinton\thanks{Purdue University, West Lafayette, IN (\email{cgb@purdue.edu})}
\and Navid Azizan\thanks{Massachusetts Institute of Technology, Cambridge, MA (\email{azizan@mit.edu})}}

\usepackage{amsopn}

\makeatletter
\newcommand*{\addFileDependency}[1]{
  \typeout{(#1)}
  \@addtofilelist{#1}
  \IfFileExists{#1}{}{\typeout{No file #1.}}
}
\makeatother

\ifpdf
\hypersetup{
  pdftitle={Distributed Linear Solvers and Data Heterogeneity},
  pdfauthor={Boris Velasevic*, Rohit Parasnis*, Christopher G.\ Brinton, Navid Azizan}
}
\fi

\begin{document}

\maketitle

\begin{abstract}
We consider the problem of solving a large-scale system of linear equations in a distributed/federated setting. The taskmaster solves the system with the help of a set of machines, each of which possesses a subset of the equations. While various solutions for this problem exist, a fundamental understanding and rigorous comparison between the convergence rates of the main algorithmic classes---the projection-based methods and the optimization-based ones---is missing. We provide the first comprehensive analysis and comparison of these two classes of algorithms, with a particular focus on the fastest representative method from each class, i.e., the Accelerated Projection-Based Consensus (APC)~\cite{azizan2019distributed} and the Distributed Heavy-Ball Method. We first introduce a novel notion of data heterogeneity called \emph{angular heterogeneity} and discuss its significance and generality. Using this notion, we characterize and compare the optimal convergence rates of the algorithms of interest and capture the effects of the number of machines, the number of equations, and cross-machine and local data heterogeneity on these rates. Our analysis sheds light on the previously observed superior performance of APC in realistic scenarios, where there is often large data heterogeneity, and provides several insights into the effect of angular heterogeneity on the efficiencies of different algorithms. Additionally, we provide distributed algorithms for efficiently computing the angular heterogeneity metrics. Lastly, as a by-product of this investigation, we obtain a tight bound on the condition number of an arbitrary matrix with full column rank in terms of the Euclidean norms of its columns and the angles between them. Numerical analyses validate our theoretical results, supporting the predicted advantage of APC in the high-heterogeneity regime and providing a deeper understanding of the effects of angular heterogeneity on convergence rates.
\end{abstract}

\begin{keywords}
distributed algorithms, iterative numerical methods for linear systems,
analysis of algorithms
\end{keywords}

\begin{MSCcodes}
68W15, 65F10, 68W40
\end{MSCcodes}

\section{Introduction}
In recent years, the proliferation of large-scale computational tasks has necessitated reliance on distributed computations across multiple processing cores or machines. This shift is driven by the complexity of the problems, which demand extensive computation and memory, or by the distributed nature of the datasets~\cite{koc1994exact,xiao2005scheme,santoro2006design,verbraeken2020survey,xiao2020survey,alimohammadi2024adapting}. Distributed implementations, in contrast to fully centralized algorithms, offer more efficient solutions with reduced computational and memory constraints and do not require moving the data to a central entity.

Of significant interest among these are distributed approaches to the problem of solving a large-scale system of linear equations. This is among the most fundamental problems in distributed computation because systems of linear equations form the backbone of innumerable algorithms in engineering and the sciences.
As a result, unsurprisingly, there exist multiple approaches for solving linear equations distributively. These can be broadly categorized as (a) approaches based on distributed optimization and (b) those specifically aimed at solving systems of linear equations.

Algorithms belonging to the former category rely on the observation that solving a linear system can be expressed as an optimization problem (a linear regression), in which the loss function is separable in the data (i.e., the coefficients) but not in the variables \cite{azizan2019distributed}. Therefore, this category includes popular gradient-based methods such as Distributed Gradient Descent (DGD) and its variants~\cite{zinkevich2010parallelized,recht2011hogwild,yuan2016convergence}, Distributed Nesterov's Accelerated Gradient Descent (D-NAG)~\cite{nesterov1983method}, Distributed Heavy-Ball Method (D-HBM)~\cite{polyak1964some}, and some recently proposed algorithms such as Iteratively Pre-conditioned Gradient-Descent (IPG)~\cite{chakrabarti2021robustness}. Besides, the Alternating Direction Method of Multipliers (ADMM)~\cite{boyd2011distributed}, a well-known algorithm that happens to be significantly slower than the others for this problem, also falls into this category.

As for the second category, i.e., approaches that are specific to solving linear systems, the most popular is the block Cimmino method~\cite{duff2015augmented,sloboda1991projection,arioli1992block}, which is a block row-projection method \cite{bramley1992row} and is essentially a distributed version of the Kaczmarz method~\cite{karczmarz1937angenaherte}. In addition, there exist some recent approaches such as those proposed in~\cite{alaviani2020distributed,huang2022scalable,mou},  and Accelerated Projection-based Consensus (APC)~\cite{azizan2019distributed}.

Among all these methods from either category, of particular interest to us are algorithms whose convergence rates are linear (i.e., the error decays exponentially in time) and whose computation and communication complexities are linear in the number of variables. These include DGD, D-NAG, D-HBM, B-Cimmino, APC, and the projection-based distributed solver proposed in~\cite{mou}. It has been shown analytically in~\cite{azizan2019distributed} that D-HBM has a faster rate of convergence than the other two gradient-based methods, i.e., DGD and D-NAG, and that APC converges faster than the other two projection-based methods, namely the block Cimmino method and the algorithm of~\cite{mou}.

However, which among the aforementioned methods is the fastest remains hitherto unknown, because it has so far proven challenging to characterize the relationships between the optimal convergence rates of the gradient-based approaches and those of the projection-based approaches. This is because the optimal convergence rates of the latter class of methods depend on the properties of the local data spaces (spaces associated with the local subsystems of equations stored at the machines), whereas the optimal convergence rates of the former class depend only on the condition number of the global system of equations.

In many distributed environments, a user may control the local data partitions by permuting the equations of the global system. For projection-based methods, this row-to-machine assignment affects the interaction between local data blocks and, in turn, the efficiency of the algorithm. Determining an optimal assignment is a formidable combinatorial problem, as the number of partitions grows factorially with the system size. This further motivates the theoretical bounds developed in this paper, which provide a framework for understanding how the distribution of data across machines affects convergence. Our characterization of the convergence rates also makes explicit how the structural differences between projection-based and gradient-based algorithms shape their performance. The resulting angular heterogeneity metrics offer qualitative insight into which algorithmic class is better suited to a given data geometry. These qualitative trends are summarized in Table~\ref{tab:heterogeneity_table}.

To circumvent this complexity, we propose a novel approach to capture the effect of the partitioning of the equations among the machines on the convergence rates of the aforementioned gradient-based and projection-based methods. Our analysis is based on a novel notion of data heterogeneity called \textit{angular heterogeneity}, which is inspired by similar notions used in the context of federated learning to quantify the diversity of local data across machines. This concept enables us to compare algorithms from both classes of interest and to show that APC converges faster than all the other methods when there is significant cross-machine angular heterogeneity, a regime we expect to arise in distributed settings with heterogeneous data placement across  machines. Consequently, a practitioner might want to partition the data so that cross-machine angular heterogeneity is maximized.

\subsection{Summary of Contributions}
Our contributions are summarized below.
\begin{enumerate}
    \item \textbf{\textit{A Geometric Notion of Data Heterogeneity.}} We propose a concept of angular heterogeneity, which extends the notion of cosine similarity/principal angles to certain vector spaces associated with the distribution of global data among the machines. Furthermore, we show that angular heterogeneity can be computed using efficient algorithms. The generality of this concept, its scale-invariant nature, and its tractability make it a potentially useful measure of data heterogeneity in distributed learning.
    \item \textbf{\textit{Convergence Rate Analysis.}}  We derive bounds on the optimal convergence rates of (a) three gradient descent-based methods, namely D-NAG, D-HBM, and DGD, and (b) three projection-based methods, namely the block Cimmino method, APC, and the algorithm proposed in~\cite{mou}. These bounds serve two different purposes: one of them captures the dependence of the optimal convergence rate on the total number of equations in the global system, whereas the other captures its dependence on the number of machines. Moreover, we show that the greater the level of cross-machine angular heterogeneity, the better the optimal convergence rates of APC and the block Cimmino method. As by-products of this analysis, we obtain bounds on the convergence rates of D-NAG, DGD, the block Cimmino method, and the algorithm of~\cite{mou}.
    \item \textbf{\textit{Illustrative Examples.}} We provide several examples to illustrate the implications of our main results for different levels of local and cross-machine angular heterogeneity.
    \item \textbf{\textit{Experimental Validation.}} We validate our theoretical results numerically and demonstrate how the optimal convergence rate of the most efficient method (APC) is bounded with respect to the dimension of the data. Our experiments also shed light on the dependence of this quantity on the number of machines. In addition, we investigate the behavior of the cross-machine angular heterogeneity in a stochastic setting involving normally distributed equation coefficients, and we show that this heterogeneity approaches its maximum value as the number of equations goes to infinity.
\end{enumerate}

\textit{Notation:} We let $\R$ denote the set of real numbers. Given two natural numbers $n$ and $m$, we let $[n]:=\{1,2,\ldots, n\}$  denote the set of the first $n$ natural numbers,  $\R^n$ denotes the space of all $n$-dimensional column vectors with real entries, and $\R^{m\times n}$ denotes the space of real-valued matrices with $m$ rows and $n$ columns. Besides, $I_{n\times n}\in\R^{n\times n}$ denotes the identity matrix and ${O}_{n\times n}\in\R^{n\times n}$ denotes the matrix with every entry equal to 0, where the subscripts are dropped if they are clear from the context.  For each $k\in[n]$, we let $e_k$ denote the $k$-th canonical basis vector of $\R^n$.

For a vector $v\in\R^n$, we let $v_k$ denote the $k$-th entry of $v$ for each $k\in [n]$, and $\norm{v}:=\sqrt{\sum_{k=1}^n v_k^2}$ denotes the Euclidean norm of $v$.  Given a matrix $M\in\R^{m\times n}$, we let $M^\top$ denote the transpose of $M$, we let $\norm{M}:=\sup_{\norm{z}=1}\norm{Mz}$ denote the spectral matrix norm of $M$,
and we let $\max(|M|)$ denote the element of $M$ with the largest absolute value. In addition, for a symmetric, nonzero, positive semi-definite matrix $M\in\R^{n\times n}$, we let $\lambda_{\max}(M)$ and $\lambda_{\min+}(M)$ denote, respectively, the maximum and the minimum nonzero eigenvalues of $M$. Analogously, let $\sigma_{\max}(M), \sigma_{\min+}(M)$ denote, respectively, the maximum and minimum nonzero singular values of a nonzero matrix $M$. Furthermore, $\kappa_+(M)$ denotes the restricted condition number of $M$ as defined with respect to the spectral norm, i.e., $\kappa_+(M):= \frac{\sigma_{\max}(M)}{\sigma_{\min+}(M)}$, while the condition number of a matrix $M$ that has strictly positive singular values is given by $\kappa(M) = \frac{\sigma_{\max}(M)}{\sigma_{\min}(M)}$. In the case of an invertible matrix $M\in\mathbb R^{n\times n}$, it is well-known~\cite{meyer2000matrix} that $\kappa(M) = \|M\|\|M^{-1}\|$. All matrix inequalities hold entry-wise.
We denote by $M^{\dagger}$ the Moore--Penrose pseudo-inverse of a matrix $M$.

Two linear subspaces $\mathcal U$ and $\mathcal V$ are said to be orthogonal if $u^\top v=0$ for all $u\in \mathcal U$ and all $v\in\mathcal V$, to express which we write $\mathcal U\perp \mathcal V$.

Finally, we let $\mathbb{P}(E)$ denote the probability of an event $E$, we let $\mathcal N(\mu,\sigma^2)$ denote  the Gaussian  distribution with mean $\mu\in\R$ and variance $\sigma^2>0$, and we let $\mathbb S^{n-1}\subset \R^n$ denote the unit sphere in $\mathbb{R}^n$.

\section{Problem Formulation}
Our goal is to compare the efficiencies of two classes of distributed linear system solvers, i.e., gradient-based algorithms and projection-based algorithms, by relating the optimal convergence rates of these methods to the degree of data heterogeneity in the network. We introduce the problem setup, describe the most efficient algorithms from each class, and reproduce some known results on their optimal convergence rates. We then introduce a few geometric notions of data heterogeneity, using which we formulate our problem precisely.

\subsection{The Setup}\label{subsec:solutions}

Consider a large-scale system of linear equations
\begin{align}\label{eq:global}
    Ax=b,
\end{align}
where $A\in\R^{N\times n}$, $x\in\R^n$, and $b\in\R^N$. Unless otherwise stated, the only assumption made on $A$ is that each of its rows contains at least one non-zero entry. A mild assumption that a subset of our results make on $A$ is the standard assumption that all its columns are linearly independent. In particular, it is standard to exclude linearly dependent features from a linear machine learning model, and we point the reader to algorithms using which this can be achieved efficiently in a distributed manner~\cite{distributedqr}.

While our primary motivation is to establish an analytical bridge between two classes of distributed linear solvers so as to enable comparison, individual algorithms often rely on restrictive assumptions (e.g., solution uniqueness, $A$ being square, etc.). We construct this bridge by deriving novel bounds on the condition numbers that dictate convergence rates. Because these quantities are fundamental matrix properties independent of any particular algorithm or solution concept, our results provide standalone value beyond the specific methods analyzed.
An important subset of these results is particularly useful when the matrix of coefficients is square and full-rank, in which case the system has a unique solution.\footnote{For instance, solving the least-squares problem $\min_{x \in \mathbb{R}^n} \frac{1}{2}\|Ax - b\|_2^2$ is mathematically equivalent to solving the normal equations $A^\top A x = A^\top b$. The matrix $A^\top A$ can efficiently be assembled distributively via $A^\top A=\sum_{i = 1}^m A_i^\top A_i$, where $A_i$ is the submatrix of $A$ stored locally at machine $i\in[m]$.} However, in general, the solution may not be unique and may even fail to exist. In this paper, we restrict our focus to cases in which a solution exists, and we are interested in finding any solution, i.e., simply finding an $x\in \R^n$ that satisfies \eqref{eq:global}.

To solve the $N$ equations specified by~\eqref{eq:global} distributively over a network of $m\le N$ edge machines, the central server partitions the global system~\eqref{eq:global} into $m$ linear subsystems as
\begin{align}\label{eq:distributed_main}
\begin{bmatrix}
A_1 \\
\hdashline
A_2 \\
\hdashline
\vdots \\
\hdashline
A_m
\end{bmatrix}
x
=
\begin{bmatrix}
b_1 \\
\hdashline
b_2 \\
\hdashline
\vdots \\
\hdashline
b_m
\end{bmatrix}
\end{align}
where for each $i\in[m]$, the $i$-th subsystem $A_ix=b_i$ (equivalently, the \textit{local data} pair $[A_i,b_i]$ where $A_i\in\R^{p_i\times n}$ and $b_i\in\R^{p_i}$) consists of $p_i$ equations and is accessible only to machine $i\in [m]$. In some applications, these local data may already be available at the respective machines without the server distributing them. A schematic representation of this setup is shown in Figure~\ref{fig:problem_setup}, which is adapted from~\cite[Figure 1]{azizan2019distributed}.
\begin{figure}[tph]
    \centering
    \includegraphics[width=0.4\columnwidth]{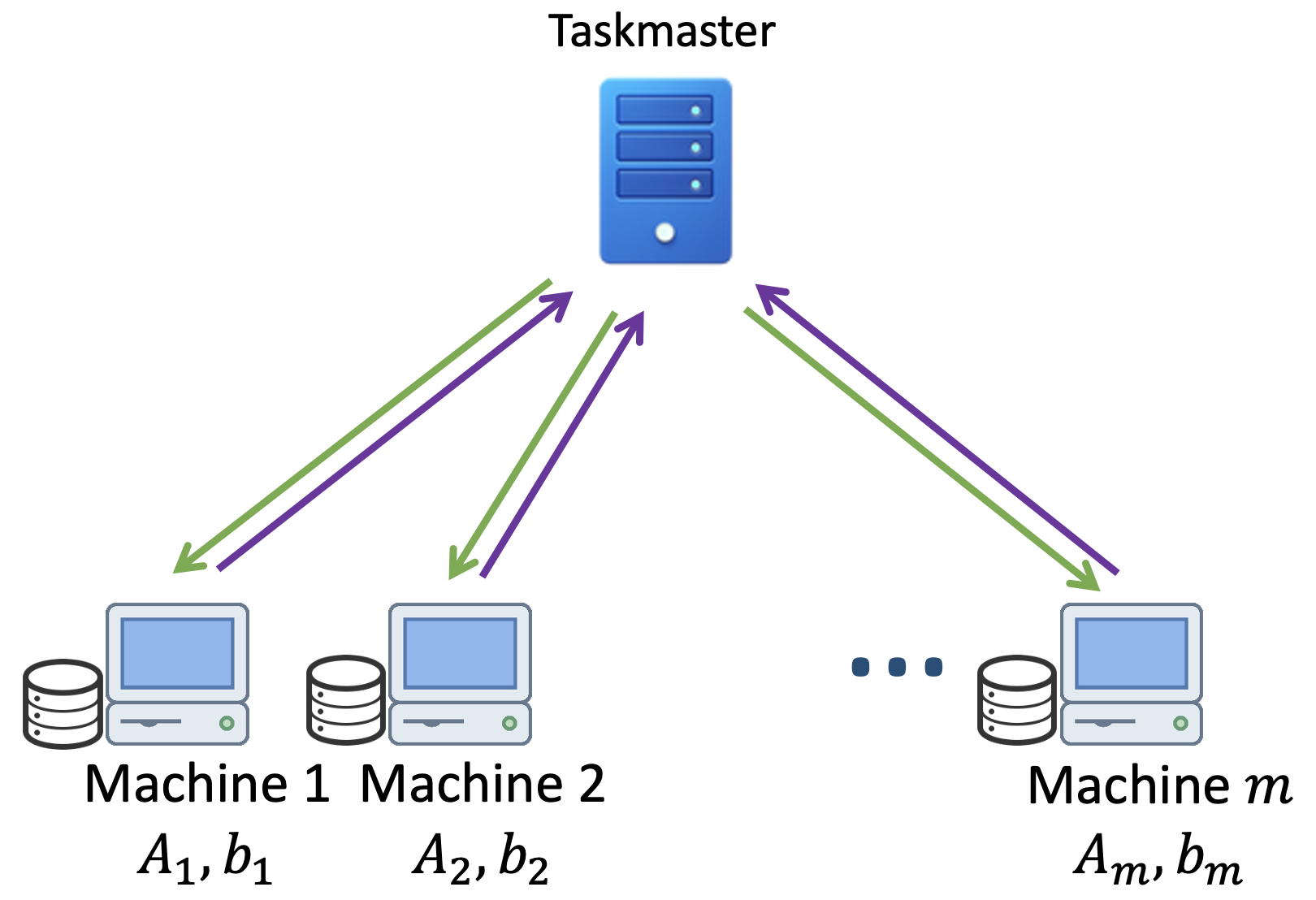}
    \captionsetup{justification=centering}
    \caption{Schematic representation of the taskmaster and the $m$ machines/cores.
Each machine $i$ has only a subset of the equations, i.e., $[A_i, b_i]$.}
    \label{fig:problem_setup}
\end{figure}

Unlike the analysis in~\cite{azizan2019distributed}, we do not impose any restrictions on the number of equations in any of these local subsystems. Note, however, that if $m\gg 1$, then each of these local systems is likely to be highly undetermined with infinitely many solutions because the number of local equations is likely much smaller than $n$.

\begin{definition}
For a fixed choice of algorithm parameters, let $x(t)$ denote the algorithm's estimate at iteration $t$, and suppose these estimates converge to a limit $x^*:=\lim_{t\to\infty}x(t)$.
The convergence error is \(\varepsilon(t) := \|x(t) - x^*\|\).
We say that the algorithm has convergence rate
\(\rho \in [0,1)\) if there exists a constant \(K \ge 0\) such that
\(
\varepsilon(t) \le K \rho^t
\)
for all sufficiently large \(t\). The optimal convergence rate of an algorithm is defined as
the infimum of all such \(\rho\) over all admissible parameter choices.
\end{definition}

\begin{remark}
    In general, the optimal convergence rate of a distributed algorithm depends on $A$, which is assumed to be a fixed, arbitrary matrix throughout the paper.
\end{remark}

We focus particularly on APC, the projection-based method with the fastest convergence behavior, and D-HBM, the gradient-based method with the fastest convergence behavior.

\subsection{Accelerated Projection-Based Consensus}\label{subsec:apc}
Originally proposed in~\cite{azizan2019distributed}, accelerated projection-based consensus (APC) is  a distributed linear system solver in which every iteration consists of a local projection-based consensus step followed by a global averaging step, both of which incorporate momentum terms that accelerate the convergence of the algorithm to the global solution of~\eqref{eq:global}.  The original formulation of APC assumes that $A$ has full row rank, which in particular forces $N \le n$. This guarantees that~\eqref{eq:global} admits a solution, and that the solution is unique when $A$ is square. We now describe  APC  in detail.

\subsubsection{The Iterates} At all times $t$, every device $i\in [m]$ stores its estimate $x_i(t)$ of the global solution $x^*$. Likewise, the server stores its estimate $\bar x(t)$ of $x^*$.

\subsubsection{ Initialization} The estimates $\{\bar x(t)\}\cup\{x_i(t):i\in[m]\}$ are initialized as follows. Each machine $i$ sets its initial estimate $x_i(0)$ of $x^*$ to a solution of $A_i x = b_i$, which can be easily computed in $O(p_i^3)$ steps. The machine then transmits $x_i(0)$ to the server, which then computes its own initial estimate of $x^*$ as $\bar x(0):=\frac{1}{m}\sum_{i=1}^m x_i(0)$. In subsequent iterations, $x_i(t)$ and $\bar x(t)$ are updated as described below.

\subsubsection{Projection-Based Consensus Step} In iteration $t+1$, every machine $i$ receives $\bar x(t)$, the server's most recent estimate of $x^*$. The machine then updates its own estimate $x_i(t)$ by performing the following projection-based consensus step:
\begin{align}\label{eq:projection_definition}
    x_i(t+1) = x_i(t) + \gamma P_i(\bar x(t) - x_i(t)),
\end{align}
where $\gamma\ge1$ is a fixed {momentum} and the matrix $P_i := I - A_i^{\dagger}A_i$ is the orthogonal projector onto the nullspace of $A_i$. This projector is well-defined for an arbitrary coefficient matrix $A$. When $A_i$ has full row rank (as assumed in the original formulation of APC~\cite{azizan2019distributed}), the pseudo-inverse admits the closed form $A_i^{\dagger} = A_i^\top(A_iA_i^\top)^{-1}$, and the expression for $P_i$ reduces to $I - A_i^\top(A_i A_i^\top)^{-1} A_i$. In other words, the machine takes an accelerated step in a direction that is orthogonal to the coefficient vectors of its local system of equations (i.e., the rows of $A_i$). This ensures that ${A_ix_i(t+1) = A_i x_i(t)=A_ix_i(0)= b_i}$, i.e., $x_i(t)$ never leaves the local solution space of machine $i$.

\subsubsection{Global Averaging Step} The next step in iteration $t+1$ is the memory-augmented global averaging step performed by the server as
\begin{align}\label{eq:glob_avg}
    \bar x(t+1) = \frac{\eta}{m} \sum_{i=1}^m x_i(t+1) + (1-\eta)\bar x(t),
\end{align}
where $\eta\in\R$ is a fixed momentum and $(1-\eta)\bar x(t)$ is the memory term.

In sum,~\eqref{eq:projection_definition} and~\eqref{eq:glob_avg} are the update steps that define APC.

\subsubsection*{Optimal Convergence Rate}
It was shown in~\cite{azizan2019distributed} that for appropriate choices of $\gamma$ and $\eta$, APC converges to a solution of~\eqref{eq:global}, which is unique when $A$ is a square matrix. The optimal convergence rate of APC depends on the arithmetic mean of the projectors $\{P_i\}_{i=1}^m$. More precisely, let
\begin{align}\label{eq:def_S}
    S & :=\sum_{i=1}^m (I - P_i)= \sum_{i=1}^m A_i^{\dagger}A_i.
\end{align}
 We then know from~\cite[Theorem 1, Theorem 3]{azizan2019distributed}\footnote{Note that $S=mX $ for the matrix $X $ defined in~\cite[Eq. (4)]{azizan2019distributed}. } that there exist values of $\gamma$ and $\eta$ that result in the optimal convergence rate of APC being given by
\begin{align}\label{eq:rho_APC}
    \rhoapc:=1-\frac{2}{\sqrt{\kappa_+\left(S\right)}+1}.
\end{align}
\def\rhobc{\rho_{\text{BC} } }
\def\rhomou{\rho_{\text{MLM} } }
\subsubsection*{Other Projection-Based Methods} It was shown in~\cite{azizan2019distributed} that the optimal convergence rate of the block Cimmino (BC) is given by ${\rhobc:=1 - \frac{2}{\kappa_+(S)+1}}$. Besides, the   algorithm proposed in~\cite{mou} has an optimal convergence rate given by ${\rhomou:= 1 - \frac{1}{m}\lambda_{\min}(S)}$, where \textit{MLM} stands for Mou, Liu, and Morse. As shown in~\cite{azizan2019distributed}, we have ${\rhoapc\le\rhobc\le\rhomou}$.

\subsection{Distributed Heavy-Ball Method}\label{subsec:dhbm}

Introduced in~\cite{polyak1964some}, the Distributed Heavy-Ball Method (D-HBM) is a distributed linear system solver that performs the following momentum-enhanced updates in each iteration $t$:
\begin{align}
    z(t+1)&=\beta z(t)+\sum_{i=1}^m A_i^\top\left(A_i x(t)-b_i\right)\\
    x(t+1)&=x(t)-\alpha z(t+1),
\end{align}
Here, {$\alpha,\beta\in\R$ are the step size and momentum parameters, respectively,} and $A_i^\top(A_ix(t)-b_i)$ is the gradient of the  function ${f_i:\R^n\to [0,\infty)}$ defined by $f_i(y)=\frac{1}{2}\norm{A_iy-b_i}^2$. This gradient is evaluated at the global estimate $x(t)$ by machine $i$. Thus, each iteration of D-HBM consists of a memory-augmented gradient update followed by an accelerated gradient descent step. Unlike APC, D-HBM imposes no structural assumptions on $A$. Importantly, the distributed architecture only parallelizes the evaluation of the local gradients $\nabla f_i$; the taskmaster aggregates them to reconstruct the full gradient. The update sequence, per-iteration complexity, and convergence rate coincide with those of the centralized Heavy-Ball Method, which matches the per-iteration complexity of APC~\cite{azizan2019distributed}.

\subsubsection*{Optimal Convergence Rate} We know from~\cite{lessard2016analysis, polyak1964some} that for proper choices of $\alpha$ and $ \beta$, D-HBM converges to a minimizer of the quadratic objective $\frac{1}{2}\|Ay-b\|_2^2$.
For the optimal choices of $\alpha$ and $\beta$, the global estimate $x(t)$ in D-HBM converges to $x^*$ as fast as $\rhodhbm^t $ vanishes, where
\begin{align}\label{eq:rho_dhbm}
    \rhodhbm :=1-\frac{2}{\sqrt{\kappa_+(A^\top A)}+1}.
\end{align}

\def\rhonag{\rho_{\text{NAG}}}
\subsubsection*{Other Gradient-Based Methods} It was shown in~\cite{lessard2016analysis, polyak1987introduction} that DGD has an optimal convergence rate given by ${\rhodgd:=1 - \frac{2}{\kappa_+(A^\top A)+1}}$, and it was shown in~\cite{lessard2016analysis, nesterov1983method} that the optimal convergence rate of D-NAG is ${\rhonag:=1 - \frac{ 2 }{\sqrt{ 3\kappa_+(A^\top A) +1
 }  }}$. Thus, ${\rhodhbm\le \rhonag\le \rhodgd}$.

 We use $\kappa_+$ in~\eqref{eq:rho_dhbm} to emphasize that the convergence rate is governed by the nonzero spectrum of $A^\top A$. In the subsequent comparison of gradient-based algorithms, however, we make the standard assumption that $A$ has linearly independent columns, which implies that  $A^\top A$ is nonsingular and $\kappa_+(A^\top A)=\kappa(A^\top A)$.

 \section{Geometric Notions of Data Heterogeneity}
In this section, we develop two geometric notions of data heterogeneity. The first notion is based on the following concepts of local data spaces and cosine similarities/principal angles between local data.

\begin{definition} [\textbf{Local Data Spaces}]
    Given a machine ${i\in[m]}$, the row space of $A_i$, {denoted by $\row(A_i)$}, is called the \textit{local data space} of machine $i$.
\end{definition}

Note that the linear span of the coefficient vectors (the rows of $A$) stored at machine $i$ equals the span of the rows of $A_i$, which is precisely the local data space of the machine.

\begin{definition} [\textbf{Principal Angles \cite{computingthetah}}]
    \label{def:principal} Let $p=\min\{\dim\mathcal U,\dim\mathcal V\}$. The principal angles $\alpha_1, \dots, \alpha_p$ between arbitrary linear subspaces $\mathcal{U}$ and $\mathcal{V}$ are recursively defined for $k = 1, \dots, p$ by:
\begin{gather}
    \cos \alpha_k
    \;=\;
    \max_{\substack{u \in \mathcal{U},\; v \in \mathcal{V} \\ u^\top u_\ell = 0,\; v^\top v_\ell = 0,\; \ell \in [k-1]}}
    \frac{u^\top v}{\|u\|\,\|v\|}
    \;=\;
    \frac{u_k^\top v_k}{\|u_k\|\,\|v_k\|},
\end{gather}
where $u_k \in \mathcal{U}$ and $v_k \in \mathcal{V}$ are the vectors achieving the maximum.
\end{definition}

It can be easily verified that $\cos \alpha_k$ is a non-increasing sequence. In this work, we are primarily focused on the smallest principal angle between the local data spaces of any two distinct machines.

\begin{definition} [\textbf{Cosine Similarity/Smallest Principal Angle}] \label{def:cosine} For any two distinct machines $i,j\in[m]$,  let $\theta_{ij}$ denote the minimum angle between their local data spaces, i.e.,
    \begin{align}\label{eq:cos}
        \theta_{ij}:=\cos^{-1}\left(\max_{u\in \row(A_i), v\in \row(A_j)}\left(\frac{| u^\top v| }{\norm{u}\norm{v}}\right)\right).
    \end{align}
    Then $\cos\theta_{ij}$ is the \emph{cosine similarity} between the local data of machines $i$ and $j$, and $\theta_{ij}$ is the \emph{smallest principal angle} between their local data spaces.
\end{definition}

\begin{remark}
Making the local data spaces $\row(A_i)$ and $\row(A_j)$ uni-dimensional in~\eqref{eq:cos} and setting $\norm{u}=\norm{v}=1$  would result in an expression for the cosine similarity between two unit-norm vectors.  Definition~\ref{def:cosine}, therefore, generalizes the standard definition of cosine similarity.
\end{remark}

On the basis of  Definition~\ref{def:cosine}, we now define cross-machine angular heterogeneity.

\begin{definition} [\textbf{Cross-machine Angular Heterogeneity}]\label{def:cross-machine}
    The inverse cosine of the maximum of all pairwise cosine similarities, i.e.,
    $$
        \theta_{\emph{H}}:=\cos^{-1}\left( \max_{1\le i< j\le m} \cos{\theta_{ij}} \right),
    $$
    is  the \emph{cross-machine angular heterogeneity} of the network.
\end{definition}

Note that we  have {$0\le\theta_{\text{H}}\le \theta_{ij}\le \frac{\pi}{2}$} for all pairs of machines $i,j\in [m]$. Also, note that the more the local data spaces of the machines diverge from each other in the angular sense, the greater is the cross-machine angular heterogeneity of the network. At the same time, however, a salient feature of this notion of data heterogeneity is that its value is invariant with respect to any scaling applied to the rows of $A$. This is especially useful in the context of solving a linear system, because the true solution of such a system is unaffected by scaling any subset of the equations.

\begin{remark}
    To understand the generality of Definition~\ref{def:cross-machine}, we revisit~\eqref{eq:global}. When $A$ is assumed to have linearly independent columns,~\eqref{eq:global} can be expressed as the unconstrained minimization of the loss function $f:\R^n\to\R$ defined by $f(x)=\frac{1}{2}\norm{Ax-b}^2$. When implemented distributively, this can be interpreted as the problem of learning a machine learning model $x$ with the rows of $A_i$ denoting the data samples stored at machine $i$ and the corresponding entries of $b$ denoting the data labels. This suggests that Definitions~\ref{def:cosine} and~\ref{def:cross-machine} can be adapted for distributed learning scenarios by replacing $\{\row(A_i)\}_{i=1}^m$ in these definitions with the linear spans of the data samples stored at the respective machines.
\end{remark}

Besides cross-machine heterogeneity, we define another geometric notion of data heterogeneity to quantify the total angular spread of the local data within each machine.
\begin{definition} [\textbf{Local Angular Heterogeneity}]\label{def:loc_het}
    The local angular heterogeneity $\phi_i$ of machine $i$ is  the minimum angle between any two of its local coefficient vectors (i.e., the rows of $A_i$).
    $$
        \phi_i :=  \cos^{-1}\left(\max_{1\le k<\ell\le p_i }\left(\frac{ |e_k^\top A_iA_i^\top e_\ell|}{\norm{A_i^\top e_k} \norm{A_i^\top e_\ell }  } \right)\right).
    $$
\end{definition}
We adopt the convention that $\phi_i = \frac{\pi}{2}$ if $p_i < 2$ and set $\phimin = \min_i \phi_i$. Finally, we define a geometric notion to quantify the angular spread of the columns of $A$, which depends only on the global system and not on how the data are distributed among the machines. This quantity is useful to analyze the algorithms that use distributed computing to parallelize matrix operations but do not operate on subsystems of the original system directly.

\begin{definition} [\textbf{Column-Wise Angular Heterogeneity}]\label{def:total_het_col}
    The column-wise angular heterogeneity $\psi$ is the minimum angle between any two columns of $A$:
    $$
        \psi :=  \cos^{-1}\left(\max_{1 \leq k < l \leq n}\left(\frac{ |e_k^\top A^\top A e_\ell|}{\norm{A e_k} \norm{A e_\ell }  } \right)\right).
    $$
\end{definition}

Note that local angular heterogeneity (respectively, column-wise angular heterogeneity) can be interpreted as the \textit{coherence} (introduced in \cite{donoho2001uncertainty}) of $A_i^\top$ (respectively, $A$).

Our goal now is (a) to analyze the effects of both local and non-local angular heterogeneity on the convergence rates $\rhoapc$ and $\rhodhbm$, and (b) to use the results of our analysis to compare the efficiencies of APC and D-HBM.

\subsection{Computability of $\thetah$}

Before discussing our main results, we point the reader to known results that enable the computation of $\thetah$ for an arbitrary collection of local data spaces.

It is not immediately clear from the definition of $\thetah$ that computing the value of $\thetah$ is tractable. However, principal angles have been studied in the literature, and the following theorem proven in \cite{computingthetah} allows for a simple polynomial-time algorithm for computing $\thetah$.
\begin{theorem}[\textbf{Principal Angle Computation~\cite{computingthetah}}]
 Recall that a matrix $M \in \mathbb{R}^{n \times p}$ can be QR-decomposed as $M = Q_MR_M$ such that the columns of $Q_M$ form an orthonormal basis for the column space of $M$. Let $Q_i$ and $Q_j$ be the orthonormal bases for the column spaces of $A_i^\top$ and $A_j^\top$, obtained for instance by rank-revealing QR. Then, $\cos(\theta_{ij})$ is the largest singular value of $Q_i^\top Q_j$.
\end{theorem}

Therefore, to compute $\thetah$, we perform Algorithm~\ref{eg:comp} below.
\begin{algorithm}{\label{algorithm1}}
\caption{Computation of $\thetah$}\label{eg:comp}
\begin{algorithmic}
\State{$\thetah \leftarrow \frac{\pi}{2}$}
\For{$1 \leq i \leq m$}
    \State{$(Q_i, R_i) \leftarrow $ {\tt QR}$(A_i^\top)$}
\EndFor
\For{$1 \leq i < j \leq m$}
    \State{$(U_{ij}, \Sigma_{ij}, V^\top_{ij}) \leftarrow $  {\tt SVD} $(Q_i^\top Q_j)$}
    \State{$\thetah \leftarrow \min(\thetah, \cos^{-1}(\max(\Sigma_{ij})))$}
\EndFor
\State{return $\thetah$}
\end{algorithmic}
\end{algorithm}

In this algorithm, {\tt QR}  is a subroutine that performs QR factorization and returns a tuple $(Q, R)$, and {\tt SVD}  is a subroutine that performs singular value decomposition and returns a tuple $(U, \Sigma, V^\top)$. If we assume that each $A_i$ has $O(p)$ rows and that $m = O({N}/{p})$, we can analyze the complexity of the algorithm as follows.  Note that QR decomposition requires $O(np^2)$  computations, which means that the complexity of the first \textit{for} loop is $O(Nnp)$.  Similarly, each matrix multiplication in the second loop has a complexity of $O(np^2)$, and the SVD decomposition has a complexity of $O(np^2)$, meaning that the complexity of the second \textit{for }loop is $O(N^2 n)$. Therefore, the complexity of the algorithm is $O(N^2 n)$, which can be prohibitively high for large values of $N, n$.

Nevertheless, we can leverage the fact that Algorithm~\ref{eg:comp} lends itself naturally to the distributed computation setting described in this paper, wherein it requires fewer computations. This gives rise to the following algorithm (Algorithm~\ref{eg:comp1}).

\begin{algorithm}{\label{algorithm2}}
\caption{Distributed computation of $\thetah$}\label{eg:comp1}
\begin{algorithmic}
\For{each machine $i \in [m]$ in parallel}
    \State{$(Q_i, R_i) \leftarrow $ {\tt QR}$(A_i^\top)$}
    \State{Send $Q_i$ to the taskmaster}
\EndFor
\State{The taskmaster provides machine $i$ with $Q_1, ..., Q_m$}
\For{each machine $i \in [m]$ in parallel}
    \State{$\theta_{i} \leftarrow \frac{\pi}{2}$}
    \For{$j \in [ m], j \neq i$}
        \State{$(U_j, \Sigma_{j}, V^\top_j) \leftarrow $  {\tt SVD}$(Q_i^\top Q_j)$}
        \State{$\theta_{i} \leftarrow \min(\theta_{i}, \cos^{-1}(\max(\Sigma_{j})))$}
    \EndFor
    \State{Send $\theta_i$ to the taskmaster}
\EndFor
\State{The taskmaster computes $\thetah \leftarrow \min\{\theta_1, ..., \theta_m\}$}
\State{return $\thetah$}
\end{algorithmic}
\end{algorithm}
In this setting, each machine first performs $O(np^2)$ computations to run Gram--Schmidt orthogonalization. Furthermore, each matrix multiplication requires $O(np^2)$ computations and each SVD requires $O(np^2)$ steps, meaning that the computational cost of computing $\theta_i$ is $O(Nnp)$. Finally, the taskmaster picks the smallest $\theta_{i}$ in linear time. As the computations of $\theta_i$ are done in parallel, the complexity of this algorithm is $O(Nnp)$.
A more detailed discussion on the numerical computation of principal angles can be found in \cite{computingthetah}.

\section{Main Results}\label{sec:main_results}
To shed light on how data heterogeneity affects the convergence rates $\rhoapc$ and $\rhodhbm$, we first bound the condition numbers $\kappa_+(S )$ and $\kappa(A^\top A)$ (with $S $ as in~\eqref{eq:def_S}) in terms of the angular heterogeneity measures {$\theta_{\text{H}}$ and $\{\phi_i\}_{i=1}^m$}. We then use~\eqref{eq:rho_APC} and~\eqref{eq:rho_dhbm} to compare $\rhoapc$ with $\rhodhbm$.

Our first result applies to the extreme case of maximum cross-machine angular heterogeneity ($\theta_{\text{H}}=\frac{\pi}{2}$). Since this is realized precisely when all the local data spaces $\{\row(A_i)\}_{i=1}^m$ are mutually orthogonal, we call this case \textit{total orthogonality}.

\begin{proposition} [\textbf{Total Orthogonality}]\label{prop:total_orth}
Suppose $\theta_{\text{H}}=\frac{\pi}{2}$, i.e., suppose $\row(A_i)\perp \row(A_j)$ for all distinct $i,j\in[m]$. Then $\kappa_{+}(S)=1$. If $A$ satisfies the assumptions required by APC (see Section~\ref{subsec:apc}), then $\rhoapc=0$; equivalently, APC converges to $x^*$ in a single step.
\end{proposition}
\begin{proof}
Recall that $I-P_i = A_i^{\dagger}A_i$ is the orthogonal projector onto $\row(A_i)$. We write $\tilde{P}_i := I-P_i$, so that $S=\sum_{i=1}^m \tilde{P}_i$. We now fix distinct indices $i,j\in[m]$. Since $\tilde{P}_i$ is the orthogonal projector onto $\row(A_i)$, it is symmetric and idempotent, its image is $\row(A_i)$, and its null space is $\ker(\tilde{P}_i) = \row(A_i)^\perp$ because the kernel of a symmetric matrix is the orthogonal complement of its image/range space. Total orthogonality states that $\row(A_i)\perp \row(A_j)$, i.e., $u^\top v=0$ for every $u\in \row(A_i)$ and $v\in \row(A_j)$. Equivalently, $\row(A_j)\subseteq \row(A_i)^\perp$. Now let $x\in\R^n$ be arbitrary. Then $\tilde{P}_j x\in \row(A_j)\subseteq \row(A_i)^\perp=\ker(\tilde{P}_i)$, and therefore $\tilde{P}_i\tilde{P}_j x = 0$. As $x$ was arbitrary, $\tilde{P}_i\tilde{P}_j = O$ for all distinct $i,j\in[m]$. Hence, we obtain
\begin{align*}
    S^2 = \sum_{i=1}^m\sum_{j=1}^m \tilde{P}_i\tilde{P}_j = \sum_{i=1}^m \tilde{P}_i^2 = \sum_{i=1}^m \tilde{P}_i = S,
\end{align*}
using $\tilde{P}_i^2=\tilde{P}_i$. Being a sum of the symmetric matrices $\tilde{P}_i$, the matrix $S$ is symmetric and hence has real eigenvalues. The identity $S^2=S$ forces every eigenvalue $\lambda$ of $S$ to satisfy $\lambda^2=\lambda$, i.e., $\lambda\in\{0,1\}$. Since the singular values of a symmetric matrix are the absolute values of its eigenvalues, every nonzero singular value of $S$ equals $1$. Moreover, $S$ has at least one non-zero singular value. This is because $O$ is the only  symmetric matrix that is nilpotent, and we have $S\ne O$ since $S u = \tilde P_i u=u\ne \mathbf 0$ for all $u\in \row(A_i)\setminus\{\mathbf 0\}$. Therefore $\sigma_{\max}(S)=\sigma_{\min+}(S)=1$, which implies $\kappa_{+}(S)=1$. This holds for arbitrary $A$. Whenever $A$ additionally meets the assumptions required by APC,~\eqref{eq:rho_APC} then gives $\rhoapc=0$, i.e., APC converges to $x^*$ in a single step.
\end{proof}

The above proposition is powerful, as it allows us to construct examples where APC converges in a single step, while D-HBM converges arbitrarily slowly. These examples are presented later in Section~\ref{sec:explicit}.

We now extend Proposition~\ref{prop:total_orth} to results that apply to a range of values of the cross-machine heterogeneity $\theta_{\text{H}}$. This results in two bounds on $\kappa_+(S)$ in terms of $\theta_{\text{H}}$, one of which is independent of the number of machines $m$ and the other is independent of the number of equations $N$.

\begin{theorem}\label{thm:kappa_S1}
Consider the distributed system~\eqref{eq:distributed_main},  and suppose there exists a machine with $p$ equations. Additionally, suppose $\theta_{\text{H}}$ satisfies $$(N - p)\cos(\thetah) + \sqrt{(N - p)^2\cos^2(\thetah) + 4(N - p)p\cos(\thetah)} < 2.$$ Then, regardless of the number of machines $m$, we have
$$
    \kappa_+(S)\le \frac{2 + (N - p)\cos(\thetah) + \sqrt{(N - p)^2\cos^2(\thetah) + 4(N - p)p\cos(\thetah)}}{2 - (N - p)\cos(\thetah) - \sqrt{(N - p)^2\cos^2(\thetah) + 4(N - p)p\cos(\thetah)}}.
$$
\end{theorem}
We provide a sketch of the proof and relegate the full proof to the supplemental materials.
\begin{enumerate}
    \item\label{step:one_basis} We express $S$ as $S= VV^\top$, where $V$ is a matrix defined in terms of the orthonormal bases of $\{\tilproj_i\}_{i=1}^m$, where $\tilproj_i:=I-P_i$ denotes the orthogonal projector onto $\row(A_i)$ for each $i\in[m]$.
    \item\label{step:two_QR} We perform a QR factorization on the matrix whose columns are the orthonormal basis vectors from the previous step.
    \item\label{step:three_comments} We comment on the properties of individual components (QR factors) obtained from the QR factorization.
    \item\label{step:four_Weyl} We use the generalized matrix determinant lemma and leverage the properties of the QR factors to bound the eigenvalues of $S$, thereby bounding the condition number.
\end{enumerate}

The bound in Theorem~\ref{thm:kappa_S1} illustrates how high levels of angular heterogeneity make the condition number of $S$ small, and consequently, the optimal convergence rate of APC better. To complement this bound, we now derive a simpler bound in terms of $m$, the number of machines.

\begin{theorem}\label{thm:kappa_S}
For any $N\times n$ system of equations and $m \geq 2$ machines with a given cross-machine angular heterogeneity $\theta_{\text{H}} > \cos^{-1}\left( \frac{1}{m - 1}\right)$, the following bound holds independent of the number of equations $N$:
$$
    \kappa_+(S )\le \frac{1 + (m - 1)\cos \theta_{\text{H}} }{1-(m - 1)\cos{\theta_{\text{H}}}}.
$$
\end{theorem}

We provide a sketch of the main innovation of the proof and relegate the full proof to the supplemental material.
\begin{enumerate}
    \item\label{step:one_basis1} We repeat the first step outlined in the proof sketch for Theorem~\ref{thm:kappa_S1}.
    \item\label{step:two_minimax} We split $V$ into its constituent block matrices $\{V_i\}_{i=1}^m$, which are defined in terms of the orthonormal bases of  the range spaces of $\{\tilproj_i\}_{i=1}^m$. We then prove $\lambda = \sum_{i = 1}^m \|V_i^\top x\|^2$ for  every eigenvalue-eigenvector pair $(\lambda, x)$ of $S$ that satisfies $\|x\| = 1, \lambda \neq 0$.
    \item\label{step:three_rephrasing} We evaluate $\lambda^2 = \|VV^\top x\|^2 = \|\sum_{i = 1}^m V_iV_i^\top x\|^2$ and  use the definition of $\theta_{\text{H}}$ to control the size of the cross-terms. In conjunction with the orthonormality of the columns of $V_i$ and conclusions from the previous step, we obtain a quadratic inequality in $\lambda$.
    \item\label{step:four_final} As $S$ is symmetric positive semi-definite, recall that the singular values of $S$ coincide with the eigenvalues of $S$. As $\lambda > 0$ we safely divide the inequality by $\lambda$, which directly yields the desired bounds on $\lambda$. This provides us with a bound on the minimum and the maximum nonzero singular values, completing the proof.
\end{enumerate}

Theorem~\ref{thm:kappa_S} shows that the condition number of $S$ (and hence also the convergence rate $\rhoapc$) is upper-bounded by an expression that is independent of the number of equations $N$. As expected, higher cross-machine angular heterogeneity leads to tighter bounds. In scenarios where datapoints stored in different machines are sampled from sources with low correlation across sources (but arbitrary within-source correlation), it is likely that APC will have superior convergence. Moreover, the result suggests that increasing the number of machines may slow down the convergence rate, which is in agreement with our intuition that packing more local data spaces into the same global space $\R^n$ may result in reducing the angular divergence between the less similar data spaces\footnote{Note that the value of $\theta_{\text{H}}$ depends only on the pair of local data spaces with the greatest (rather than least) cosine similarity.}.

Similarly, Theorem~\ref{thm:kappa_S1} provides a bound on the condition number that is independent of the number of machines $m$, but that gets worse as $N$ increases. Again, this is in line with our observations as seen in the Experiments section.

Having examined $\kappa_+(S )$, which determines $\rhoapc$, we now examine $\kappa(A^\top A)$, which determines $\rhodhbm$.

\begin{theorem}\label{prop:a_transpose_a}
Assume that $A$ has linearly independent columns. Let $a_k^\top\in\R^{ 1\times n}$ denote the $k$-th row of $A$ for each $k\in [N]$, and let $\mathcal{S}$ be any $n$-element subset of $[N]$ such that $\{a_k : k \in \mathcal{S}\}$ are linearly independent. Denote by $A_{\mathcal{S}}$ the sub-matrix of $A$ obtained by restricting $A$ to the rows indexed by $\mathcal{S}$. We then have
\begin{align}\label{eq:final_a_tr_a}
    \kappa(A_{\mathcal{S}}) \ge \frac{\max_{k\in {\mathcal{S}}} \norm{a_k} }{  \min_{\ell\in {\mathcal{S}} }\left\{ \norm{a_\ell} \sin\thetamin^{(\ell)}\right\} },
\end{align}
where
$$
    \thetamin^{(\ell)}:=\cos^{-1}\max_{k\in {\mathcal{S}} \setminus\{\ell\} }\left(\frac{|a_k^\top a_\ell  |}{ \norm{a_k } \norm{a_\ell} } \right)
$$
is the minimum angle between $a_\ell^\top$ and any other row in $\{a_k : k \in \mathcal{S} \setminus \{\ell\} \}$. In the special case when $n = N$, setting $\mathcal{S} = [n]$ yields a bound on the condition number of $A$ in terms of the angles between its rows.
\end{theorem}

We outline the key steps of the proof below and relegate the full proof to the supplemental material.
\begin{enumerate}
    \item\label{step:one} As described above, we identify a submatrix $A_{\mathcal{S}}$ of $A$ that is square and invertible. As $A_{\mathcal{S}}$ is square, its singular values are the same as those of $A_{\mathcal{S}}^\top$.
    \item\label{step:two}  We derive the bound $\sigma_{\max}(A_{\mathcal{S}}^\top) \geq \max_{k \in \mathcal{S}} \|a_k\|$ using the Courant-Fischer theorem for singular values.
    \item\label{step:three}We derive the bound $\sigma_{\min}(A_{\mathcal{S}}^\top) \leq \min_{\ell \in \mathcal{S}} \{\|a_\ell\| \sin(\theta^{(\ell)}_{\min})\}$ using the Courant-Fischer theorem for singular values.
    \item\label{step:four} We integrate the results of the previous three steps and complete the proof.
\end{enumerate}

The interpretation of Theorem~\ref{prop:a_transpose_a} is as follows. We might be interested in solving a ``core" system of $n$ independent equations out of the $N$ given equations if the remaining $N-n$ equations do not provide additional information, i.e., if they are linearly dependent on the ``core" $n$ equations. Furthermore, in the case where $A$ is a square matrix by itself, we can derive a lower bound on the condition number solely as a function of row-angular quantities.

If we use angles between column vectors instead of angles between row vectors, we can adapt Theorem~\ref{prop:a_transpose_a} so that it applies to any matrix $A$ and not just to the square sub-matrices of $A$.

\begin{theorem}\label{prop:a_transpose_a_col}
Assume that $A$ has linearly independent columns. Let $\tilde{a}_k \in\R^{ N \times 1}$ denote the $k$-th column of $A$ for each $k\in [n]$. We then have
\begin{align}\label{eq:final_a_tr_a_col}
    \kappa (A) \ge \frac{\max_{k\in [n]} \norm{\tilde{a}_k} }{\min_{\ell\in [n] }\left\{ \norm{\tilde{a}_\ell} \sin\gamma_{\min}^{(\ell)}\right\} },
\end{align}

where
$$
    \gamma_{\min}^{(\ell)}:=\cos^{-1}\max_{k\in [n] \setminus\{\ell\} }\left(\frac{|\tilde{a}_k^\top \tilde{a}_\ell  |}{ \norm{\tilde{a}_k } \norm{\tilde{a}_\ell} } \right)
$$
is the minimum angle between $\tilde{a}_\ell^\top$ and any other column of $A$.

\end{theorem}

The proof of Theorem~\ref{prop:a_transpose_a_col} closely parallels that of Theorem~\ref{prop:a_transpose_a}. The restriction to a square submatrix $A_S$ in the former result was merely a device to analyze the singular values of $A_S^\top$ and to infer column independence from row independence. Under the assumptions of Theorem~\ref{prop:a_transpose_a_col}, we already have linearly independent columns, allowing us to apply  the arguments used in the proof of Theorem \ref{prop:a_transpose_a} directly to $A$.

We now comment on how Theorems~\ref{prop:a_transpose_a} and~\ref{prop:a_transpose_a_col} help us derive simple bounds on $\kappa(A^\top A)$. First, assume that $A$ is a square matrix with linearly independent columns. If \\$\ell^* = \argmin_{\ell \in [n]} \theta_{\min}^{(\ell)}$ is the row that makes the globally smallest angle, observe that\\ $\min_{\ell \in [n]}\{\|a_\ell\|\sin \theta_{\text{min}}^{(\ell)}\} \le \norm{a_{\ell^*}} \sin \theta_{\min}^{(\ell^*)}$. Moreover, we need to have that $\theta_{\min}^{(\ell^*)} \leq \phimin$ as the global minimum $\theta_{\min}^{(\ell^*)}$ has to be upper bounded by any local minimum $\phi_i$. Using these two facts, we deduce from~\eqref{eq:final_a_tr_a} that
\begin{align}\label{eq:phimin}
        \kappa(A^\top A) &\ge \left(\frac{\max_{k \in [n]}\|a_k\|}{\min_{\ell \in [n]}\{\|a_\ell\|\sin \theta_{\text{min}}^{(\ell)}\}}\right)^2 \ge \left(\frac{\max_{k \in [n]}\|a_k\|}{\|a_{\ell^*}\|\sin \theta_{\text{min}}^{(\ell^*)}}\right)^2 \ge \frac{1}{\sin^2 \theta_{\text{min}}^{(\ell^*)}} \geq \frac{1}{\sin^2 \phimin}.
\end{align}
Thus, it suffices to have just one machine with low local data heterogeneity for the condition number of $A^\top A$ to be large.

Theorem~\ref{prop:a_transpose_a} provides a bound on $\kappa(A^ \top A) = (\kappa(A))^2$ not only in terms of the minimum local angular heterogeneity  ${\phimin=\min_{i\in[m]} \phi_i}$, but also in terms of the variation in the norms of the rows of $A$.

To see the dependence on the variation in the norms of $\{a_k:k\in[n]\}$, one can easily verify that~\eqref{eq:final_a_tr_a} implies that
\begin{align}\label{eq:norm_bound}
    \kappa(A^\top A)\ge \frac{\max_{k\in [n]} \norm{a_k}^2  }{\min_{\ell\in [n]} \norm{a_\ell}^2}.
\end{align}
Therefore, a single row of $A$ with an atypically large (or small) norm suffices to make $\kappa(A^\top A)$ large.

Besides, it is worth noting that the tightness of the bound established in Theorem~\ref{prop:a_transpose_a} is evident from Example~\ref{eg:exampleone} to follow.  Furthermore, Theorems~\ref{prop:a_transpose_a} and~\ref{prop:a_transpose_a_col} are general results on condition numbers. Hence, these results may be of independent interest to the reader.

We now drop the assumption that $A$ is a square matrix. Similarly, if $\ell^* = \argmin_{\ell \in [n]} \gamma_{\min}^{(\ell)}$ indexes one of the two columns that minimize the angle between any pair of columns of $A$, we can analogously conclude that

\begin{align}\label{eq:psi}
        \kappa(A^\top A) \ge \frac{1}{\sin^2(\psi)}
\end{align}
and derive the column-oriented version of~\eqref{eq:norm_bound}:
\begin{align}\label{eq:norm_bound_col}
        \kappa(A^\top A) \ge \frac{\max_{k\in [n]} \norm{\tilde a_k}^2  }{\min_{\ell\in [n]} \norm{\tilde a_\ell}^2}.
\end{align}

Theorem~\ref{prop:a_transpose_a_col} leads to the following lower bound on $\kappa(S)$ in the crucial case of square, invertible matrices.

\begin{proposition}\label{cor}
Assume that $N = n$ and that $A$ is full-rank. We always have $\kappa(S) \ge \frac{1}{\sin^2\theta_{\text{H}}}$ independently of the number of equations $n$ and the number of machines $m$.
\end{proposition}

\begin{proof}
    By definition, there exist two machines $i,j\in[m]$ and two unit-norm vectors $u\in\row(A_i)$ and $v\in\row(A_j)$ such that $u^\top v=\cos\theta_{\text{H}}$. We now construct a matrix $V\in\R^{n\times N}$ as described in Step~\ref{step:one_basis} in the proof of Theorem~\ref{thm:kappa_S1} after setting $v_1^i:=u$ and $v_1^j:=v$.

    Let $\tilde{N} = \sum_{i = 1}^m \tilde{p}_i$, where $\tilde{p}_i \leq p_i$ is the dimension of $\row(A_i)$ and the number of columns in $V_i$. Note that $\tilde{N} \leq N$. On the other hand, we note that the vectors $\{v^i_j : j \in [\tilde{p}_i], i \in [m]\}$ span $\mathbb{R}^n$. Therefore, $\tilde{N} \geq n$. As $n = N$, we conclude that $\tilde{N} = N = n$.

    We have thus shown that the range space of $V$ is $\R^n$. Since row rank  equals column rank, this implies that $V^\top$ has $n$ linearly independent rows and $n$ linearly independent columns.

    This observation enables us to replace $A$ with  $V^\top$, set $\mathcal{S} = [n]$ in~(\ref{eq:phimin}), and use the fact that all the rows of $V^\top$ have the same norm to obtain the desired bound.
\end{proof}

Next, we analyze the practical implications of our bounds on $\kappa(A^\top A)$ and $\kappa(S)$ under varying regimes of data heterogeneity in the square invertible case. Here, the derivation of~\eqref{eq:phimin} expresses the convergence-rate bounds in terms of $\phimin$, a row-oriented, \emph{local} measure of heterogeneity that lets us track how the distribution of data across machines affects both method classes. Although Theorem~\ref{prop:a_transpose_a} shows $\phimin$ could be replaced by the global metric $\min_{\ell\in\mathcal{S}}\theta_{\min}^{(\ell)}$, the local metric is preferable in practice, as each machine can compute it on its own. When $A$ is not square but has linearly independent columns, we can instead use the column-oriented metric $\psi$; however, Proposition~\ref{cor} no longer applies, so we lose the lower bounds on $\rhomou$, $\rhobc$, and $\rhoapc$. Being strictly global and agnostic to row shuffling, column metrics have no natural local analogue. Consequently, the gradient-based bounds are largely insensitive to how data are partitioned across machines, whereas any nontrivial bound for projection-based methods requires a row-based analysis that accounts for local data.

We now combine the bound established in  Theorem~\ref{thm:kappa_S} with the expression for $\rhoapc$ provided in~\eqref{eq:rho_APC} and simplify the result in order to obtain an upper bound on $\rhoapc$. Similarly, combining Proposition~\ref{cor} with~\eqref{eq:rho_APC} results in a lower bound on $\rhoapc$. We repeat these steps with the closed-form expressions we provided in Section~\ref{subsec:apc} for the optimal convergence rates $\rhomou$ and $\rhobc$  to obtain similar bounds, which we summarize in Table~\ref{tab:my_label}.

Likewise, we combine the bounds established in~\eqref{eq:phimin} and~\eqref{eq:norm_bound} with the expression for $\rhodhbm$ provided in~\eqref{eq:rho_dhbm} in order to obtain the following lower bounds on $\rhodhbm$:
$$
    \rhodhbm = 1 - \frac{2}{\sqrt{\kappa(A^\top A)}+1} \stackrel{(a)}\ge 1 - \frac{2}{\frac{1}{\sin\phimin}+1 } = \frac{2}{1+\sin\phimin}-1,
$$
and
$$
    \rhodhbm = 1 - \frac{2}{\sqrt{\kappa(A^\top A)}+1} \stackrel{(b)}\ge 1 - \frac{2}{\frac{\max_k \|a_k\| }{\min_\ell \|a_\ell\| }+1 } = \frac{\max_k\norm{a_k} -\min_\ell \norm{a_\ell} }{\max_k\norm{a_k} + \min_\ell\norm{a_\ell}},
$$
where $(a)$ follows from~\eqref{eq:phimin} and the fact that $\phi_i\in[0,\frac{\pi}{2}]$ for all $i\in[m]$, and $(b)$ follows from~\eqref{eq:norm_bound}. We repeat these steps with the closed-form expressions we provided in Section~\ref{subsec:dhbm} for $\rhodgd$ and $\rhonag$ to obtain similar convergence rate bounds, which we summarize in Table~\ref{tab:my_label}.

\begin{table}[h!]
    \caption{Summary of our bounds on the optimal convergence rates of projection-based and gradient-based methods. MLM: Mou, Liu, and Morse~\cite{mou}; BC: block Cimmino method~\cite{duff2015augmented,sloboda1991projection,arioli1992block};  APC: Accelerated Projection-Based Consensus~\cite{azizan2019distributed}; DGD: Distributed Gradient Descent~\cite{yuan2016convergence}; D-NAG: Distributed Nesterov's Accelerated Gradient Descent~\cite{nesterov1983method}; D-HBM: Distributed Heavy-Ball Method~\cite{polyak1964some} }
    \centering
\begin{tabular}{|| c ||}
\hline\\
    $
        1- \sin^2\thetah\le \rhomou \le \left(1 -\frac{1}{m} \right)(1+\cos\thetah)
    $
    \\\\
    \hline\\
    $
        \frac{2}{1+ \sin^2\theta_{\text{H}} }-1\le \rhobc\le (m-1)\cos\thetah
    $
    \\\\
    \hline \\
    $
        \frac{2}{1 + \sin\theta_{\text{H}}} -1 \le \rhoapc \le \frac{(m-1)\cos\theta_{\text{H} } }{ 1 + \sqrt{1- (m-1)^2\cos^2\theta_{\text{H}}}}
    $
    \\\\
    \hline\hline\\
    $   \rhodgd\ge
        \frac{\max_k\norm{a_k}^2 -\min_\ell \norm{a_\ell}^2 }{\max_k\norm{a_k}^2 + \min_\ell\norm{a_\ell}^2}
    $,
    $
        \rhodgd\ge
        \frac{2 }{1+\sin^2\phimin }-1
    $
    \\\\
    \hline\\
    $
        \rhonag \ge 1 - \frac{ 2}{\sqrt{ 3\big(\frac{ \max_k\norm{a_k} }{ \min_\ell \norm{ a_\ell}  }\big)^2 +1}}
    $,
    $
        \rhonag \ge 1 - \frac{ 2\sin\phimin }{\sqrt{3 + \sin^2\phimin}}
    $
    \\\\
    \hline\\
    $   \rhodhbm\ge
        \frac{\max_k\norm{a_k} -\min_\ell \norm{a_\ell} }{\max_k\norm{a_k} + \min_\ell\norm{a_\ell}}
    $,
    $
        \rhodhbm\ge
        \frac{2 }{1+\sin\phimin }-1
    $
    \\\\
    \hline
\end{tabular}
    \label{tab:my_label}
\end{table}

\subsection{Comparison of $\rhoapc$ and $\rhodhbm$}\label{subsec:apc_vs_dhbm}

To compare $\rhoapc$ with $\rhodhbm$ in settings with different levels of local and cross-machine angular data heterogeneity, we first deduce from~\eqref{eq:rho_APC} and~\eqref{eq:rho_dhbm} that $\rhoapc\le\rhodhbm$ if and only if $\kappa(S )\le\kappa(A^\top A)$. Hence, we obtain the following result as an immediate consequence of Theorem~\ref{thm:kappa_S} and~\eqref{eq:phimin}.

\begin{corollary}
    A sufficient condition for $\rhoapc < \rhodhbm$ is $\sin^2\phimin < \frac{1 - (m-1)\cos\theta_{\text{H}}}{1 + (m-1)\cos\theta_{\text{H}}}$, or equivalently,
\begin{align}\label{eq:sufficient_condition}
    \frac{2}{1 + \frac{1}{(m - 1)\cos\theta_\text{H}}} < \cos^2\phimin.
\end{align}
\end{corollary}

We now consider four realistic cases for $\phimin$ and $\theta_{\text{H}}$.

\subsubsection{Small $\phimin$ and Large $\theta_\text{H}$} This is the case of high cross-machine heterogeneity accompanying low local heterogeneity. Hence, we expect this regime to arise in distributed settings where different machines are exposed to different environments, producing significant variation in data across machines rather than within any local dataset. From~\eqref{eq:sufficient_condition} and the preceding discussion, it is clear that APC is likely to outperform D-HBM.

\subsubsection{Large $\phimin$ and Small $\theta_{\text{H}}$} This may happen in federated learning scenarios in which the local data, while diverse enough, are similar across different machines.
Consequently, if the global data are diverse, then so is every local dataset. This may lead to the value of $\phimin$ being large. On the other hand, since the local datasets are similar to the global dataset, they are also similar to each other.  This may result in a small $\theta_{\text{H}}$. In light of Proposition~\ref{cor}, this means that APC converges slowly in this case. At the same time, however, highly diverse global data are likely to result in a large variation in the norms of ${\{a_k:k\in[n]\}}$ (the rows of $A$). This leads to a large $\kappa(A^\top A)$, and consequently, a poor convergence rate for D-HBM too.

\subsubsection{Small $\phimin$ and Small $\theta_{\text{H}}$}
This scenario is likely to occur in settings where both local and cross-machine data exhibit a low level of heterogeneity. In this case, \eqref{eq:phimin} implies that D-HBM converges slowly. On the other hand, Proposition \ref{cor} implies that APC also converges slowly. Therefore, it is unclear which of the two algorithms performs better in this scenario.

\subsubsection{Large $\phimin$ and Large $\theta_{\text{H}}$} This scenario is likely to arise in settings where the data are highly diverse, both across devices and within each device.
As previously noted, APC converges quickly in this case. Even though both local and cross-machine heterogeneity are large, it is of note that D-HBM may still converge slowly if row norms differ significantly, as a consequence of Equation \eqref{eq:norm_bound}.

While some cases are inconclusive, APC is favored over D-HBM in the high cross-machine heterogeneity regime (Case 1), which we expect to be common in distributed and federated settings where different devices are exposed to different environments. Moreover, as Theorems~\ref{thm:kappa_S} and~\ref{prop:a_transpose_a} suggest, APC has the added advantage of its convergence rate being insensitive to any diversity in the Euclidean norms of the coefficient vectors (the rows of $A$).

Table \ref{tab:heterogeneity_table} qualitatively summarizes how the convergence rates of optimization-based methods compare with those of projection-based methods under different heterogeneity conditions.

\begin{table*}[ht]
  \centering
  \captionsetup{justification=centering}
  \caption{A summary of the performance of projection-based and optimization-based methods for different local and cross-machine heterogeneity regimes}
  \begin{tabular}{|c|}
    \hline
    \includegraphics[width=0.76\textwidth]{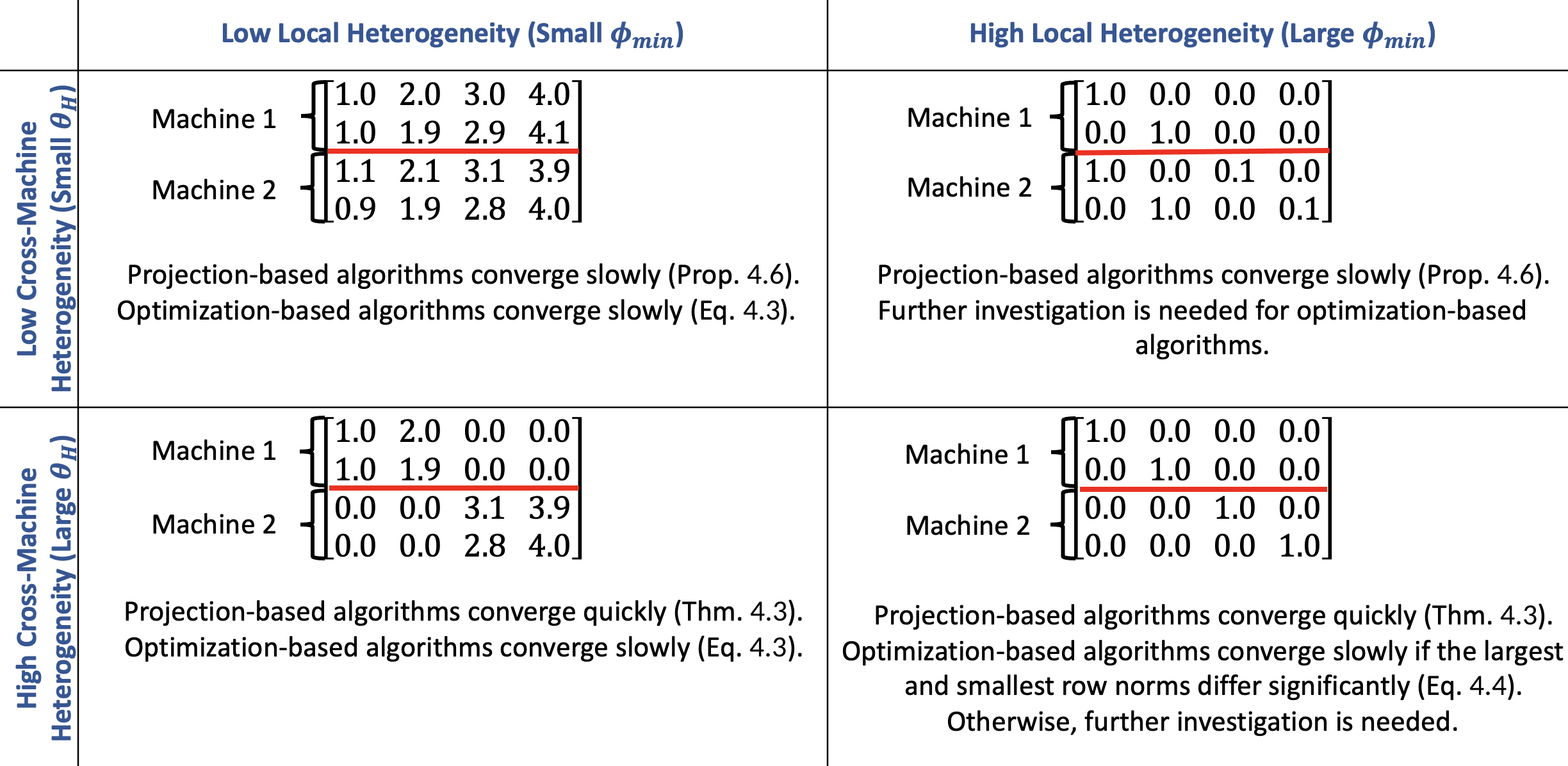} \\
    \hline
  \end{tabular}
  \label{tab:heterogeneity_table}
\end{table*}

\subsection{Comparison of Other Optimal Convergence Rates}
From the definitions of $\rhobc$ and $\rhodgd$, it is clear $(m - 1)\cos \theta_\text{H} \leq \frac{\cos^2 \phimin}{1 + \sin^2 \phimin}$ is a sufficient condition for the block Cimmino method to converge faster than DGD. Therefore, the block Cimmino method can be compared with DGD in a similar manner in which we compared APC with D-HBM in Section~\ref{subsec:apc_vs_dhbm}. Moreover, even though~\eqref{eq:sufficient_condition} is in general not applicable to comparisons made between other pairs of algorithms, we infer from Table~\ref{tab:my_label} that any comparison made between a projection-based method and a gradient-based method will be qualitatively similar to the preceding comparisons.

\section{Illustrative Examples}\label{sec:explicit}

We now present examples of systems of equations to demonstrate how the convergence rates of the two families of methods compare for each example. All examples assume $A$ to be square and invertible for ease of exposition.

\begin{example}\label{eg:exampleone}
\normalfont Consider a block-diagonal matrix ${A = \text{diag}(A_1, ..., A_k)}$, and suppose as usual that machine $i$ stores $[A_i, b_i]$ for each $i\in [k]$. Then it can be easily verified that the resulting setup manifests total orthogonality~\ref{prop:total_orth}. However, note that $$\kappa(A^\top A) = \kappa(\text{diag}(A_1^\top A_1, ..., A_k^\top A_k)).$$ Therefore, as long as $A_i$s do not all satisfy $A_i^\top A_i = cI$ with a common $c$, APC will converge strictly faster than D-HBM. Note that this conclusion holds more generally for any matrix $A$ whose rows and columns can be relabeled to transform $A$ into a block-diagonal matrix, i.e., whenever there exists a permutation matrix $P\in\R^{n\times n}$ such that $P^T A P$ is block-diagonal.

This scenario represents a machine learning setting in which $k$ sensors measure pairwise disjoint subsets of features, so that (after suitable row and column permutations) the combined measurements form a block-diagonal matrix.

A special case of this scenario occurs when ${A = \text{diag}(a_1, ..., a_n)}$ is a diagonal matrix. Here, we additionally infer that any assignment of equations to machines yields an optimal performance for APC, while the convergence rate of D-HBM may be arbitrarily poor because $\kappa(A^\top A) = \left(\frac{\max(a_{ii} )}{\min(a_{ii} )}\right)^2$ for such a matrix. Furthermore, this example serves as an equality case for Theorems~\ref{thm:kappa_S1},~\ref{thm:kappa_S},~\ref{prop:a_transpose_a},~\ref{prop:a_transpose_a_col}.
\end{example}

\begin{example}\label{eg:exampletwo}
\normalfont Denote the local data matrix of the $i^\text{th}$ out of $m = \frac{n}{p}$ machines by $A_i = \begin{bmatrix}A_{i1}, \dots, A_{im}\end{bmatrix}$ where each $A_{ij} \in \mathbb{R}^{p \times p}$.
Let the entries of $A_{ij}$ for $j \neq i$ be independent  and identically distributed (i.i.d.) as $\mathcal{N}(0, \sigma^2)$,  and let the entries of $A_{ii}$ be generated  i.i.d. according to $\mathcal{N}(0, 1)$. This example generalizes Example~\ref{eg:exampleone} by allowing the off-diagonal entries to be non-zero. Using Algorithm~\ref{eg:comp}, we can construct empirical distributions for $\thetah$ by calculating $\thetah$ for $T = 10^4$ matrices that are generated randomly as described in this example. The results are displayed in Figure~\ref{fig:empirical} for various values of $\sigma$, and we additionally plot the cutoff determined by $\cos \thetah < \frac{1}{m - 1}$, which is necessary for Theorem~\ref{thm:kappa_S} to apply.

\begin{figure}[h]
    \centering
    \includegraphics[width=0.4\columnwidth]{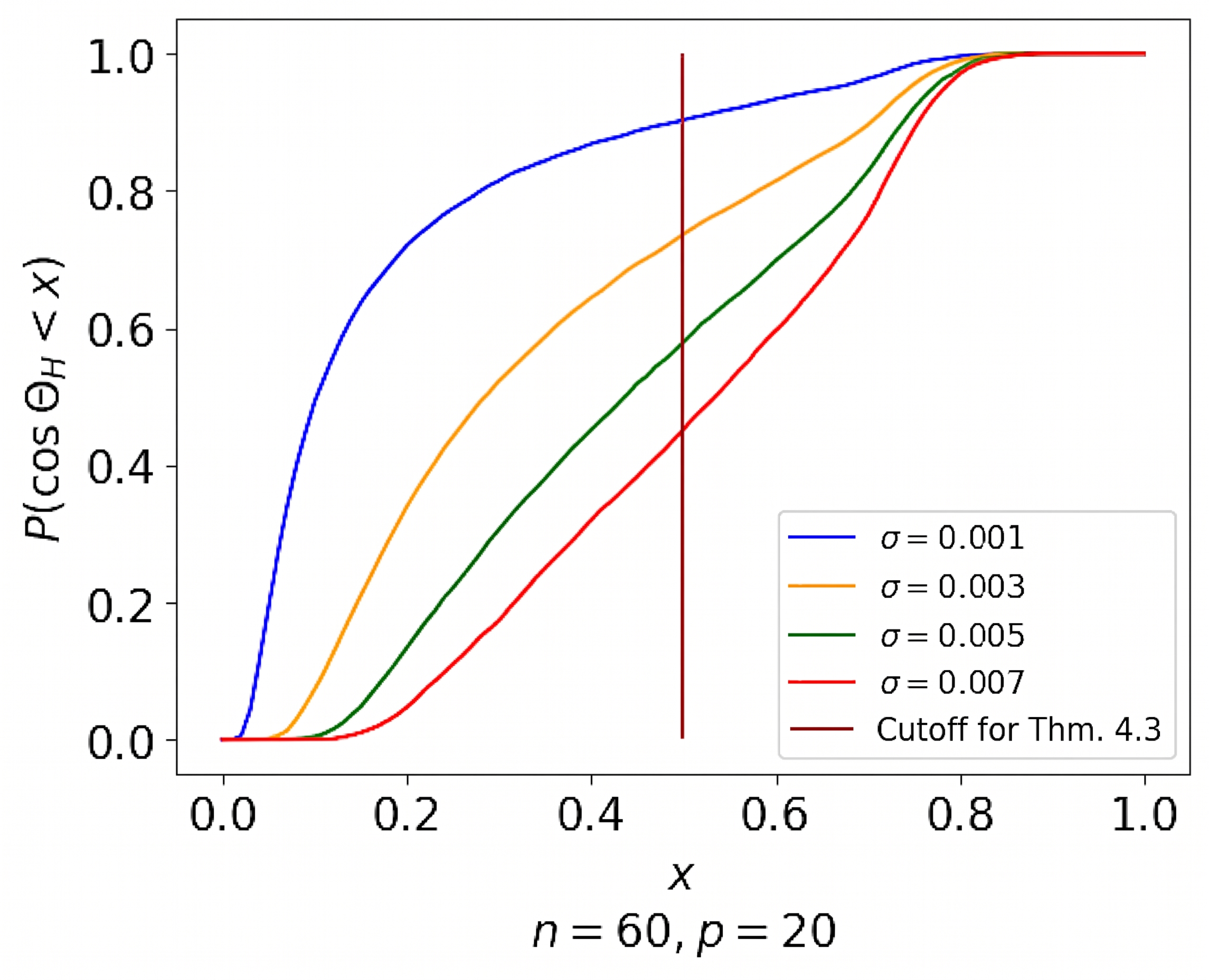}
    \captionsetup{justification=centering}
    \caption{Empirical distribution of $\cos\thetah$ for the scenario modeled in Example~\ref{eg:exampletwo}.}
    \label{fig:empirical}
\end{figure}
Note that if $\sigma$ is small, then the assumptions of Theorem~\ref{thm:kappa_S} are met, and this in turn results in an informative bound on $\thetah$.

\end{example}

\begin{example}\label{eg:examplethree}
\normalfont Let $J_n$ be the $n\times n$ matrix where every entry is $1$. Consider $A = J_n - \varepsilon I_{n}$. Assume the partition contains at least one local block with two distinct rows. For small values of $\varepsilon$, this is a matrix with little variation in its entries. Analytically calculating $S$ for a given partition in this case is tedious. Therefore, we deduce the convergence of APC through angular quantities. Note that if $a_i, a_j$ are distinct rows in $A$, we have  $a_i^\top a_j = n - 2\varepsilon$, and $\|a_i\|_2 = \sqrt{n - 2\varepsilon + \varepsilon^2}$.
Also, recall that $\theta_{ij}$ is defined as the minimum angle between any two vectors in $\mathcal{R}(A_i)$ and $ \mathcal{R}(A_j)$. Therefore, any feasible angle serves as an upper bound on $\theta_{ij}$. Consequently, the angle made by $a_i \in \mathcal{R}(A_i)$ and $a_j \in \mathcal{R}(A_j)$ is $\cos^{-1} \left(\frac{|a_i^\top a_j|}{\|a_i\|\|a_j\|}\right)$. Using our aforementioned inferences, we conclude that
$$
        \theta_{ij} \leq \cos^{-1}\left(\frac{n - 2\varepsilon}{n - 2\varepsilon + \varepsilon^2}\right) = \cos^{-1}\left(1 - \frac{\varepsilon^2}{n - 2\varepsilon + \varepsilon^2}\right).
$$
It now follows from the definition of $\thetah$ that $\thetah \leq \cos^{-1}(1 - \frac{\varepsilon^2}{n - 2\varepsilon + \varepsilon^2})$. In other words,  $\thetah$ can be chosen to be arbitrarily close to $0$. Therefore, $\kappa(S) \geq \frac{1}{\sin^2 \thetah}$ is unbounded, making the convergence of APC arbitrarily slow.

Finally, we have $\phimin = \cos^{-1}(1 - \frac{\varepsilon^2}{n - 2\varepsilon + \varepsilon^2})$, meaning that $\phimin$ can be chosen to be arbitrarily close to $0$. As a result, $\kappa(A^\top A) \geq \frac{1}{\sin^2 \phimin}$ is unbounded, thereby making the convergence of D-HBM arbitrarily slow. Therefore, when there is little to no variation in the entries of $A$, it is unclear which algorithm performs better, as both of them are inefficient.
\end{example}

\section{Experiments}\label{Experiments}

In this section, we demonstrate the empirical performance of our theoretical results through three sets of Monte-Carlo experiments.  All experiments assume $A$ to be square and invertible.
\begin{enumerate}
    \item\label{step:experiment1} In the first experiment, we keep $n$, the number of equations in the global system~\eqref{eq:global}, fixed, and we investigate how the convergence rate of APC compares with that of D-HBM for a given number of machines $m$.
    \item\label{step:experiment2} We keep the number of machines fixed and investigate how the convergence rate of APC compares with that of D-HBM.
    \item\label{step:experiment3} We investigate the asymptotic behavior of $\thetah$ in a randomized setting.
\end{enumerate}
We now describe the experiments in detail. In every experiment, we set $p_i = \frac{n}{m}$, where $n$ is the number of equations and $m$ is the number of machines.

\subsection*{Experiment~\ref{step:experiment1}: Dependence of $\rho_{\text{APC}}$ and $\rhodhbm$ on $m$} We set $n=120$ and generate multiple independent realizations of $A \in \mathbb{R}^{n \times n}$, whose entries are  i.i.d.\ random variables generated according to $\mathcal{N}(\mu, \sigma^2)$ with $\mu=0$ and $\sigma=1$. We then compute the condition numbers of $S $ and $A^\top A$. To make our simulations stable, we drop the samples where $\kappa(A^\top A) > 10^7$. Nevertheless, we make sure to obtain ${T=300}$ samples, and we compute the empirical expectations of $\rhodhbm$ through $\kappa(A^\top A)$ (which is independent of the number of machines) and $\rho_{\text{APC}} = \rho_{\text{APC}}(m)$ for each divisor $m \geq 2$ of $n = 120$.

Next, we repeat all of the above steps with $\mu=1$ to examine the phenomenon of large-mean distributions leading to more pronounced differences between the convergence rates of APC and D-HBM, as described in \cite{azizan2019distributed}. Figure \ref{fig:experiment1} plots the results of Experiment \ref{step:experiment1} for $\mu \in \{0, 1\}$.
\begin{figure}[h]
    \centering
    \includegraphics[width=0.4\columnwidth]{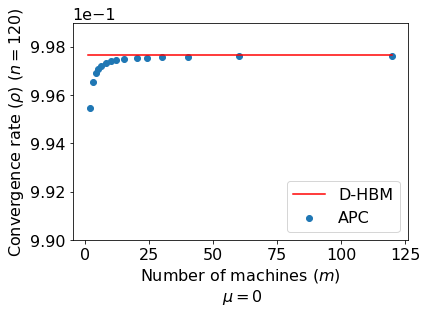}
    \includegraphics[width=0.4\columnwidth]{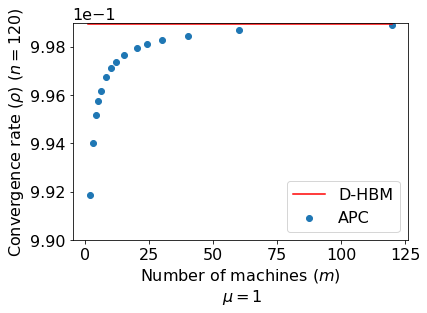}
    \captionsetup{justification=centering}
    \caption{Variation of  $\rhoapc$ and $\rhodhbm$ with $m$ for $n=120$ and $\mu\in\{0,1\}$.}
    \label{fig:experiment1}
\end{figure}

\subsubsection*{Key Inferences}
\begin{enumerate}
    \item APC clearly outperforms D-HBM in both cases. This validates the conclusions drawn from our main results in Section~\ref{sec:main_results}.
    \item As the number of machines increases, the optimal convergence rate of APC deteriorates and approaches that of D-HBM. This is to be expected for the following reason: as we increase $m$, we increase the number of local data spaces being packed into $\R^n$ (the universal data space), possibly reducing the angles between some of the local spaces. This ultimately reduces the cross-machine angular heterogeneity with some positive probability and leads to an increase in the expected value of the upper bound on $\kappa(S )$ established in Theorem~\ref{thm:kappa_S}.
    \item APC converges faster when the coefficient mean $\mu$ is increased. This is consistent with the findings of~\cite{azizan2019distributed}, and can potentially be explained by the distributional formula for $\thetah$ for non-central distributions presented in Section 10.6.4 in \cite{multivariatestats}.
\end{enumerate}

\subsection*{Experiment~\ref{step:experiment2}: Dependence of $\rho_{\text{APC}}$ and $\rhodhbm$ on $n$} We retain the setup of Experiment 1, except that we now vary the number of equations (i.e., the size of the matrix $A$) and keep the number of machines fixed, and we now draw the entries of $A$ only from $\mathcal{N}(1, 1)$.
Note that $\rhodhbm$ depends on the matrix size $n\times n$ via $\kappa(A^\top A)$. Therefore, we now expect this convergence rate to vary in our Monte-Carlo simulations.  Figure \ref{fig:experiment2} displays the results of Experiment \ref{step:experiment2} for $m \in \{10, 20\}$.
\begin{figure}[h]
    \centering
    \includegraphics[width=0.4\columnwidth]{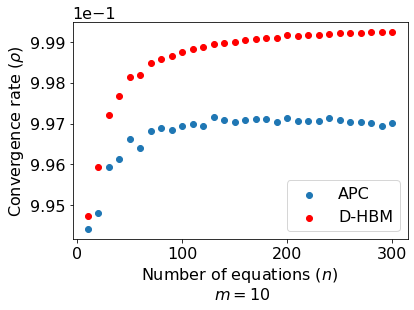}
    \includegraphics[width=0.4\columnwidth]{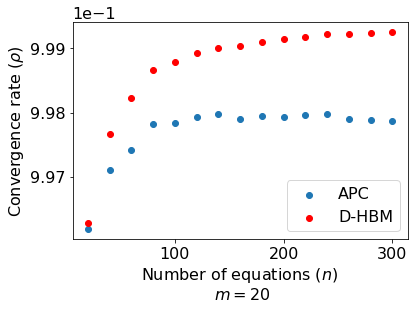}
    \captionsetup{justification=centering}
    \caption{Variation of $\rhoapc$ and $\rhodhbm$ with $n$ for fixed $m$.}
    \label{fig:experiment2}
\end{figure}

\subsubsection*{Key Inferences}
\begin{enumerate}
    \item Both $\rhoapc$ and $\rhodhbm$ increase with $n$ because the inherent complexity of~\eqref{eq:global} increases with the number of equations.
    \item The convergence rate of APC is remarkably insensitive to $n$ for large values of $n$. This can be explained with the help of Theorems~\ref{thm:kappa_S} and~\ref{prop:a_transpose_a} as follows: we know from Theorem~\ref{thm:kappa_S} that $\kappa(S )$ is upper-bounded by a quantity that depends on $n$ only through the cross-machine angular heterogeneity $\theta_{\text{H}}$. Given that the rows of $A$ are i.i.d.\ Gaussian random vectors, we do not expect $\theta_{\text{H}}$ to decrease with $n$, which implies that $\kappa(S )$ (and hence $\rhoapc$) is bounded with respect to $n$.
\end{enumerate}

\subsection*{Experiment~\ref{step:experiment3}: Asymptotic behavior of $\theta_H$}
Suppose the entries of the matrix $A$ are sampled independently from $\mathcal{N}(0, 1)$. Let $A_i, A_j \in \mathbb{R}^{p \times n}$ be two arbitrary local coefficient matrices. Consider the corresponding local data spaces $\mathcal{R}(A_i), \mathcal{R}(A_j)$ of dimension $p$ spanned by the rows of these machines, which are zero-mean Gaussian random vectors. We now argue that $\mathcal{R}(A_i)$ and $ \mathcal{R}(A_j)$ are uniformly sampled and independent  $p$-dimensional subspaces of $\mathbb{R}^n$. Note that a linear subspace of $\mathbb{R}^n$ is determined by its basis elements. The basis elements of $\mathcal{R}(A_i)$ and $ \mathcal{R}(A_j)$ are vectors whose entries are sampled i.i.d. from $\mathcal{N}(0, 1)$, and hence, the independence of such subspaces is obvious. Next, it is well-known that the row space of a $p \times n$ random matrix whose entries are i.i.d. and jointly Gaussian is uniformly distributed over the set of all $p$-dimensional subspaces of $\mathbb{R}^n$ \cite{uniformgrassmanian}.

Probability distributions of principal angles between two uniformly and independently sampled $p$-dimensional subspaces of $\mathbb{R}^n$ are well-studied in the literature \cite{asymptotics1} \cite{multivariatestats}. The exact formula for the cumulative distribution function (CDF) of $\tan^{-2}\theta_{ij}$ is given in \cite[Section 10.6.4]{multivariatestats} for the case of $\frac{n - 2p - 1}{2}$ being an integer:
\begin{align}\label{eq:formula}
    &\mathbb{P}(\tan^{-2} \theta_{ij} \leq x) = \left(1 - \frac{1}{1 + x}\right)^{\frac{p^2}{2}}\sum_{k = 0}^{\frac{p(n - 2p - 1)}{2}}\sum_{\iota}\frac{\left(\frac{p}{2}\right)_\iota C_\iota(I_p)}{(1+x)^k k!},
\end{align}
Here the sum is taken over all partitions $\iota = (\iota_1, ..., \iota_p)$ such that $\sum_{i = 1}^p \iota_i = k$, $\iota_i \geq 0 \text{ for all } i \in [p]$, and $\frac{n - 2p - 1}{2} \geq \iota_1 \geq ... \geq \iota_p$. Furthermore, $(\frac{p}{2})_\iota$ and $C_\iota$ are the hypergeometric coefficients and zonal polynomials with respect to partition $\iota$, respectively.

We now turn to numerical estimates of the probability density of $\tan^{-2}\theta_{ij}$. We keep $p$ fixed and observe how the density changes as $n$ increases. The results are plotted in Figure~\ref{fig:density1}.

\begin{figure}[h]
    \centering
    \includegraphics[width=0.4\columnwidth]{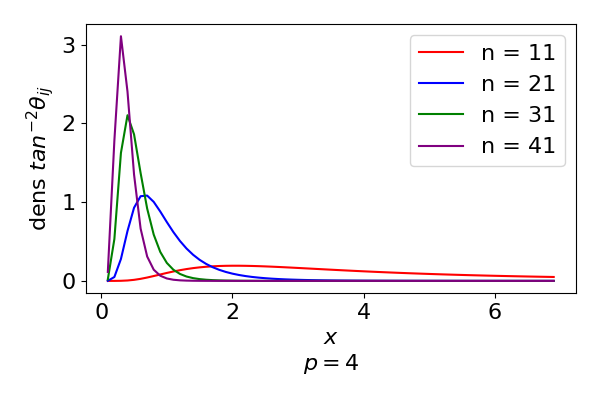}
    \includegraphics[width=0.4\columnwidth]{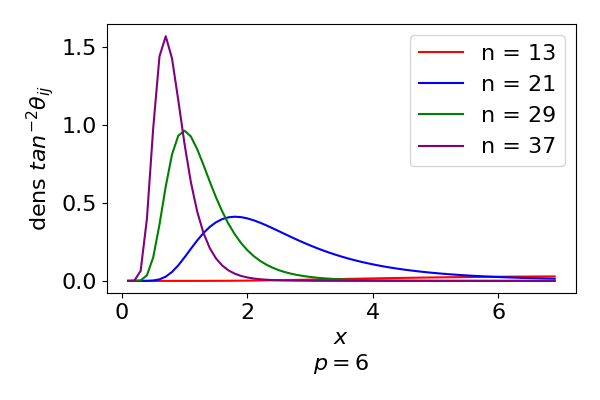}
    \captionsetup{justification=centering}
    \caption{The behavior of the density of $\tan^{-2} \theta_{ij}$ under increasing $n$ when $p \in \{4, 6\}$}
    \label{fig:density1}
\end{figure}
To estimate the behavior of $\thetah$ given the CDF of $\theta_{ij}$ for every pair $1 \leq i < j \leq m$, we employ the union bound to obtain
\begin{align*}
    \mathbb{P}(\tan^{-2} \thetah > \varepsilon) = \mathbb{P}(\exists \,(i, j): \tan^{-2}(\theta_{ij}) > \varepsilon)
    \leq \sum_{1 \leq i < j \leq m} \mathbb{P}(\tan^{-2} \theta_{ij} > \varepsilon) \cr
    = \frac{n}{2p}\left(\frac{n}{p}-1\right) \mathbb{P}(\tan^{-2} \theta_{ij} > \varepsilon),
\end{align*}
where $\varepsilon > 0$ is an arbitrary number that represents how tight of a concentration we are interested in.

We therefore investigate the asymptotics of ${\mathbb{P}(\tan^{-2} \theta_{ij} > \varepsilon)}$ as $n \to \infty$, again using a computational approach by calculating the quantity using~\eqref{eq:formula}. See Figure~\ref{fig:density2} for the results.

\begin{figure}[h]
    \centering
    \includegraphics[width=0.4\columnwidth]{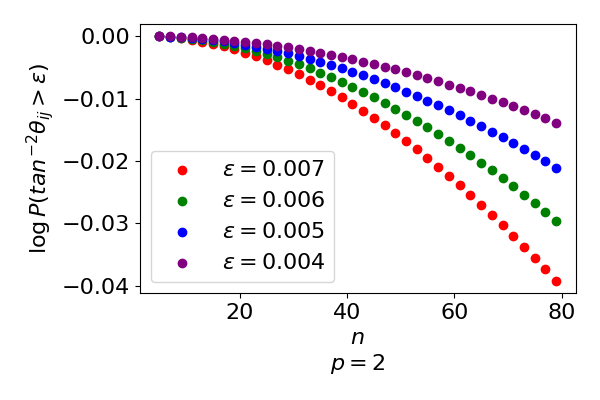}
    \includegraphics[width=0.4\columnwidth]{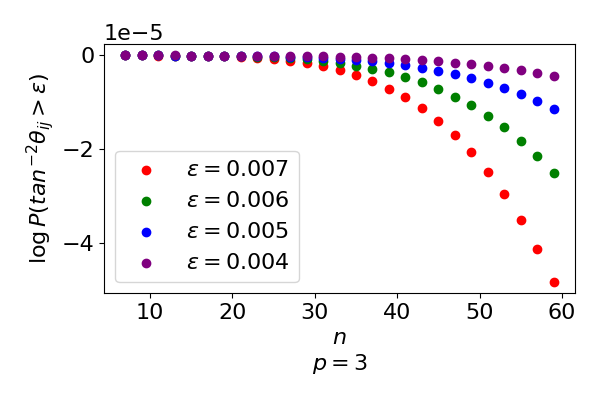}
    \captionsetup{justification=centering}
    \caption{Asymptotics of $\mathbb{P}(\tan^{-2} \theta_{ij} > \varepsilon)$ under decreasing $\varepsilon$ when $p \in \{2, 3\}$}
    \label{fig:density2}
\end{figure}

\subsubsection*{Key Inferences}
\begin{enumerate}
    \item Figure~\ref{fig:density1} suggests that as $n \to \infty$ and we keep $p$ fixed, the density of $\tan^{-2} \theta_{ij}$ concentrates around 0, meaning that $\theta_{ij} \to \frac{\pi}{2}$.
    \item Figure~\ref{fig:density2} suggests that $\mathbb{P}(\tan^{-2} \theta_{ij} > \varepsilon)$ is decreasing either exponentially or super-exponentially in the limit as $n \to \infty$. Recalling the union bound calculation, this implies that $\mathbb{P}(\tan^{-2} \thetah > \varepsilon) \to 0$ for any fixed $\varepsilon > 0$. Hence, as $n\to\infty$, we have $ \thetah \to \frac{\pi}{2}$. However, this convergence is slow and generally insufficient by itself to make definitive conclusions from Theorem~\ref{thm:kappa_S}, as it follows from the required conditions that we must have $\mathbb{P}(\tan^{-2} \thetah > \frac{p^2}{n(n - 2p)}) \to 0$ to draw such conclusions.
\end{enumerate}

The implementation of Algorithm \ref{eg:comp} and the calculations in \eqref{eq:formula} that were used to generate Figures \ref{fig:density1} and \ref{fig:density2} can be found at \href{https://github.com/azizanlab/DataHeterogeneity}{https://github.com/azizanlab/DataHeterogeneity}.

\section{Conclusion and Future Directions}

In the context of distributively solving a large-scale system of linear equations, we compared the convergence rates of two classes of algorithms that differ in their design, namely, gradient descent-based methods such as D-HBM, and projection-based methods such as APC. In doing so, we developed a novel, geometric notion of data heterogeneity called angular heterogeneity, and we used it to characterize the convergence rates of several well-known distributed linear system solvers belonging to each class. As a result of our analysis, we not only established the advantage of APC in the high cross-machine heterogeneity regime, supported by both our theoretical bounds and numerical experiments, but we also provided several insights into the effect of angular heterogeneity on the efficiencies of the studied methods. Besides, we used known results in numerical linear algebra to obtain distributed algorithms for the efficient computation of our proposed angular heterogeneity metrics. Finally, as a by-product of our investigation, we obtained a bound on the condition number of an arbitrary matrix with full column rank in terms of the Euclidean norms of its columns and the angles between them. As we have demonstrated a clear equality case in Example~\ref{eg:exampleone}, this bound is tight.

Future work could investigate the effect of the number of machines on the convergence rates of projection-based methods, such as APC. Furthermore, analyzing the smallest principal angle may provide a deeper understanding of APC in a randomized setting where the entries are Gaussian, or come from some other distribution. Other future avenues include the continual/incremental setting as in \cite{min2022one} as well as characterizing the implicit bias of APC in the underdetermined (overparameterized) case \cite{azizan2018stochastic,azizan2021stochastic}.

\bibliographystyle{siamplain}
\bibliography{references}

\end{document}


\maketitle

\section{Proofs}

\subsection{Proof of Theorem~\ref*{thm:kappa_S1}}
\begin{proof}
    We follow the steps outlined below Theorem~\ref*{thm:kappa_S1}.
    \subsubsection*{Step~\ref*{step:one_basis}}
    For each $i\in[m]$, let $\tilde P_i:=I-P_i$ denote the orthogonal projector onto $\row(A_i)$, the rowspace of $A_i$. Using the property~\cite[Eq. (5.13.4)]{meyer2000matrix1} of orthogonal projectors, we can observe that ${\tilproj_i = \sum_{j = 1}^{\tilde{p}_i} v_j^i(v_j^i)^\top}$, where $\{v_j^i:j\in [\tilde{p}_i]\}$ form an orthonormal basis for $\row(A_i)$ and $\tilde{p}_i \leq p_i$ is the dimension of $\row(A_i)$. It now follows from~(\ref*{eq:def_S}) that
    \begin{align*}
        S=\sum_{i=1}^m (I-P_i) = \sum_{i=1}^m \tilde P_i =\sum_{i=1}^m\sum_{j=1}^{\tilde{p}_i} v_j^i (v_j^i)^\top.
    \end{align*}
    Equivalently, $S=VV^\top$, where
    \begin{align*}
        V := \begin{bmatrix}
        \uparrow & \uparrow & ... & \uparrow & \uparrow & ... & \uparrow & ... & \uparrow\\
        v_1^1 & v_2^1 & ... & v_{\tilde{p}_1}^{1} & v_1^2 & ... & v_{\tilde{p}_2}^{2} & ... & v_{\tilde{p}_m}^m \\
        \downarrow & \downarrow & ... & \downarrow & \downarrow & ... & \downarrow & ... & \downarrow
        \end{bmatrix}.
    \end{align*}
    It follows that
    \begin{align}\label{eq:step_1}
        \kappa_{+}(S) = \kappa_{+}(VV^\top)\stackrel{(a)}=\frac{\lambda_{\max}(VV^\top) }{\lambda_{\min+}(VV^\top)},
    \end{align}
    where $(a)$ holds because the symmetry of $VV^\top $ implies that its eigenvalues equal its singular values. Thus, it suffices to bound the minimum and the maximum eigenvalues of $VV^\top$. Let $\tilde{N} = \sum_{i = 1}^m \tilde{p}_i$.

    \subsubsection*{Step~\ref*{step:two_QR}}
    We obtain the QR factorization $V=QR$ using Gram--Schmidt orthogonalization. Note that the first $\tilde{p}_1$ columns of $V$ constitute an orthonormal basis, which implies that $Q$ is of the form $Q = \begin{bmatrix} v_1^1 & ... & v_{\tilde{p}_1}^1 & w_{\tilde{p}_1 + 1} & ... & w_{n}\end{bmatrix}$ where $w_{\tilde{p}_1+1},\ldots, w_n\in\R^n$ and that $R$ is of the form
    $$
        R = \begin{bmatrix} I_{\tilde{p}_1\times \tilde{p}_1} & U_1 \\ 0_{(n - \tilde{p}_1) \times \tilde{p}_1} & U_2 \end{bmatrix}
    $$
    where $U_1\in \R^{\tilde{p}_1\times (\tilde{N}-\tilde{p}_1)}$ and $U_2\in \R^{(n-\tilde{p}_1)\times (\tilde{N}-\tilde{p}_1)}$. By expanding the product $QR$, we observe that
    $$
        U_0 := \begin{bmatrix}
        \uparrow & ... & \uparrow & ... & \uparrow\\
        v_1^2 & ... & v_{\tilde{p}_2}^{2} & ... & v_{\tilde{p}_m}^m \\
        \downarrow & ... & \downarrow & ... & \downarrow
        \end{bmatrix} = Q\begin{bmatrix} U_1 \\ U_2 \end{bmatrix}.
    $$
    \subsubsection*{Step~\ref*{step:three_comments}}  We first comment on the component $U_1$ of $R$. From QR factorization, we obtain  $\begin{bmatrix} V_1 & V_2 & \hdots & V_m \end{bmatrix} = QR = \begin{bmatrix} V_1 & W \end{bmatrix} R$, where
    $$
      V_k := \begin{bmatrix}
        \uparrow & \uparrow & ... & \uparrow \\
        v_1^k & v_2^k & ... & v_{\tilde{p}_k}^{k}  \\
        \downarrow & \downarrow & ... & \downarrow
        \end{bmatrix},
        \quad
        W:= \begin{bmatrix}
        \uparrow  & ... & \uparrow \\
        w_{\tilde p_1+1} &  ... & w_n  \\
        \downarrow  & ... & \downarrow
        \end{bmatrix}.
    $$
    By inspecting element-wise, we obtain that $\begin{bmatrix} V_2 & \hdots & V_m \end{bmatrix} = V_1U_1 + WU_2$. We now left-multiply by $V_1^\top$. Note that $V_1^\top V_1 = I$ as the columns of $V_1$ are orthonormal. As columns of $Q$ are orthonormal, we get that $V_1^\top W = 0$. Therefore $\begin{bmatrix} V_1^\top V_2 & \hdots & V_1^\top V_m \end{bmatrix} = U_1$. By the definition of $\theta_{1i}$, we have $|(v_k^1)^\top v_{\ell}^i| \leq \cos\theta_{1i}$, as $v_k^1$ and $ v_{\ell}^i$ are in the orthonormal bases of $\mathcal{R}(A_1)$ and $ \mathcal{R}(A_i)$, respectively. Therefore, $\max(|U_1|) \leq \max_{2 \leq i \leq m} \cos\theta_{1i}$ and $\max(|U_1^\top U_1|) \leq \tilde{p}_1 \max_{2 \leq i \leq m} \cos\theta_{1i} \leq p_1 \max_{2 \leq i \leq m} \cos\theta_{1i}$.

    Next, we comment on the relationship between $U_1, U_2$. Observe that there exist matrices $M_1^2,\ldots, M_{m-1}^2,M_1^3,\ldots M_{m-1}^3, M_1^m, M_2^m,\ldots$ such that
    \begin{align}
        \begin{bmatrix} I_{\tilde{p}_2\times \tilde{p}_2} & M_1^2 & ... & M_{m - 1}^{2}  \\ M_1^3 & I_{\tilde{p}_3\times \tilde{p}_3} & ... & M_{m - 1}^3 \\ & \ddots & & \\ M_1^{m} & M_2^m & ... & I_{\tilde{p}_m \times \tilde{p}_m} \end{bmatrix} = U_0^\top U_0 = U_1^\top U_1 + U_2^\top U_2.
    \end{align}
    Note that we can write $ U_1^\top U_1 + U_2^\top U_2 = I + L$, where $L$ contains all the off-diagonal entries. As $L$ is obtained by taking inner products between unit-norm basis vectors across different machines, using the same reasoning as above gives $\max(|L|) \leq \max_{2 \leq i < j \leq m}\cos\theta_{ij}$.  Furthermore, as $U_1^\top U_1 + U_2^\top U_2$ and $I$ are both symmetric, it follows that $L$ is also symmetric.
    \subsubsection*{Step~\ref*{step:four_Weyl}} We now turn to bounding the non-zero eigenvalues of $S = VV^\top$. Note that $VV^\top = QRR^\top Q^\top$, and $RR^\top$ have the same eigenvalues. Indeed, observe that
    \begin{align*}
        \det(QRR^\top Q^\top - \lambda I) &= \det(QRR^\top Q^\top - \lambda QQ^\top)\cr
        &= \det(Q(RR^\top - \lambda I) Q^\top)\cr
        &= \det(Q)^2 \det(RR^\top - \lambda I)\cr
        &= \det(RR^\top - \lambda I),
    \end{align*}
    where we used the fact that $Q$ is an orthogonal matrix. Therefore, we are interested in bounding the non-zero eigenvalues of
    $$
        R R^\top = \begin{bmatrix} I_{\tilde{p}_1 \times \tilde{p}_1} + U_1U_1^\top & U_1 U_2^\top \\ U_2 U_1^\top  & U_2U_2^\top \end{bmatrix} = \begin{bmatrix} I & 0 \\ 0 & 0 \end{bmatrix} + \begin{bmatrix} U_1 \\ U_2 \end{bmatrix}  \begin{bmatrix} U_1^\top & U_2^\top \end{bmatrix}.
    $$

    To that end, let $\lambda$ be an arbitrary non-zero eigenvalue of $RR^\top$. As $RR^\top$ is positive semi-definite (because $y^\top RR^\top y = (R^\top y)^\top (R^\top y)=\|R^\top y\|^2\ge 0$ for all $y$), we know that $\lambda \geq 0$, and since $\lambda\ne 0$, we conclude that $\lambda > 0$. Furthermore, assume that $\lambda \neq 1$. We must then have
    \begin{align*}
        0 &= \det(RR^\top - \lambda I)\cr
        &= \det\bigg(\begin{bmatrix} I - \lambda I & 0 \\ 0 & -\lambda I\end{bmatrix} + \begin{bmatrix} U_1 \\ U_2 \end{bmatrix}  \begin{bmatrix} U_1^\top & U_2^\top \end{bmatrix}\bigg)  \\ &= \det\bigg( I + \begin{bmatrix} U_1^\top & U_2^\top \end{bmatrix}\begin{bmatrix} \frac{1}{1 - \lambda} & 0 \\ 0 & \frac{1}{-\lambda} \end{bmatrix}\begin{bmatrix} U_1 \\ U_2 \end{bmatrix} \bigg) \det\bigg(\begin{bmatrix} (1 - \lambda) I & 0 \\ 0 & -\lambda I \end{bmatrix}\bigg) \\ &= \det\bigg( I + \begin{bmatrix} U_1^\top & U_2^\top \end{bmatrix}\begin{bmatrix} \frac{1}{1 - \lambda} & 0 \\ 0 & \frac{1}{-\lambda}  \end{bmatrix}\begin{bmatrix} U_1 \\ U_2 \end{bmatrix} \bigg) (1 - \lambda)^{\tilde{p}_1}\lambda^{n - \tilde{p}_1}(-1)^{n-\tilde p_1}
    \end{align*}
    where the third equality holds by the generalized matrix determinant identity and that $\lambda \notin \{0, 1\}$. It follows that:
    $$
    0 = \det\bigg(I + \frac{1}{1 - \lambda} U_1^\top U_1 - \frac{1}{\lambda} U_2^\top U_2\bigg).
    $$
    Substituting $U_2^\top U_2 = I + L - U_1^\top U_1$ (which follows from the definition of $L$), we obtain:
    $$
    0 = \det\bigg(\frac{\lambda - 1}{\lambda} I - \frac{1}{\lambda} L + \frac{1}{\lambda(1 - \lambda)} U_1^\top U_1\bigg).
    $$
    Multiplying by $\lambda^{\tilde{N} - \tilde{p}_1}$,  we obtain:
    $$
    0 = \det \bigg((\lambda - 1) I - L - \frac{1}{\lambda - 1} U_1^\top U_1 \bigg).
    $$
    Therefore $\lambda - 1$ is an eigenvalue of $L + \frac{1}{\lambda - 1} U_1^\top U_1$. Let the associated unit eigenvector be $x$. From there, we obtain $\lambda - 1 = x^\top L x + \frac{1}{\lambda - 1} x^\top U_1^\top U_1 x$. As $\max(|L|) \leq \max_{2 \leq i < j \leq m} \cos(\theta_{ij}) \leq \cos(\thetah)$, by  Gerschgorin circle theorem the spectrum of $L$ is bounded by $(\tilde{N} - \tilde{p}_1) \cos(\thetah)$. As $L$ is symmetric and $\|x\| = 1$, by the Courant--Fischer theorem this also serves as a bound on $|x^\top L x|$. Similarly, as $\max(|U_1^\top U_1|) \leq p_1 \max_{2 \leq i \leq m} \cos(\theta_{1i}) \leq p_1 \cos(\thetah)$, we attain that the spectrum of $U_1^\top U_1$ is bounded by $(\tilde{N} - \tilde{p_1})p_1\cos(\thetah)$. As $U_1^\top U_1$ is symmetric and $\|x\| = 1$, this also serves as a bound on $|x^\top U_1^\top U_1 x|$. Finally, note that $\tilde{N} - \tilde{p}_1 = \sum_{i = 2}^m \tilde{p}_i \leq \sum_{i = 2}^m p_i = N - p_1$. This gives us bounds on $|x^\top L x|, |x^\top U_1^\top U_1 x |$ of $(N - p_1)\cos(\thetah), (N - p_1)p_1 \cos(\thetah)$.

    By a straightforward application of triangle inequality, we conclude that $|\lambda - 1| \leq (N - p_1)\cos(\thetah) + \frac{1}{|\lambda - 1|}(N - p_1)p_1 \cos(\thetah) \Rightarrow |\lambda - 1|^2 \leq |\lambda - 1|(N - p_1) \cos(\thetah) + (N - p_1)p_1\cos(\thetah)$. Treating this as a quadratic in $|\lambda - 1|$, an elementary calculation yields bounds on any eigenvalue $\lambda \neq 1$ to be
    \begin{align*}
    &\lambda \in \mathcal{I}, \text{ where } \\ & \mathcal{I} = [1 - \frac{(N - p_1)\cos(\thetah) + \sqrt{(N - p_1)^2\cos^2(\thetah) + 4(N - p_1)p_1\cos(\thetah)}}{2}, \\
    & \quad\quad 1 + \frac{(N - p_1)\cos(\thetah) + \sqrt{(N - p_1)^2\cos^2(\thetah) + 4(N - p_1)p_1\cos(\thetah)}}{2}]
    \end{align*}
    As $1 \in \mathcal{I}$, we conclude that for any non-zero eigenvalue $\lambda$ of $S$, we have that $\lambda \in \mathcal{I}$. Given our assumptions, we can divide the right and left endpoints of the interval to attain our conclusion. As the choice of machine index was arbitrary, we can replace $p_1$ by any $p_i$ for $i\in [m]$.
\end{proof}

\subsection{Proof of Theorem~\ref*{thm:kappa_S}}
\begin{proof} We perform the steps outlined below Theorem~\ref*{thm:kappa_S}.
    \subsubsection*{Step~\ref*{step:one_basis1}}
    This step is identical to the first step in the proof of Theorem~\ref*{thm:kappa_S1}.

    \subsubsection*{Step~\ref*{step:two_minimax}}  Define $$
        V_i := \begin{bmatrix}
        \uparrow & \uparrow & ... & \uparrow \\
        v_1^i & v_2^i & ... & v_{\tilde{p}_i}^{i} \\
        \downarrow & \downarrow & ... & \downarrow
        \end{bmatrix}
    $$
    so that $V = \begin{bmatrix}V_1 & \cdots & V_m \end{bmatrix}$ and $S = VV^\top = \sum_{i = 1}^m V_iV_i^\top$. Let $(\lambda, x)$ be any eigenvalue-eigenvector pair of $S$ that satisfies $\|x\| = 1, \lambda \neq 0$. Observe that $\lambda x = VV^\top x = \sum_{i = 1}^m V_iV_i^\top x$. By left-multiplying by $x^\top$, we obtain that $\lambda = \sum_{i = 1}^m \|V_i^\top x\|^2$.

    \subsubsection*{Step~\ref*{step:three_rephrasing}}

    By evaluating $\|VV^\top x\|^2$, we obtain  $$\lambda^2 = \|VV^\top x\|^2 = \bigg\|\sum_{i = 1}^m V_iV_i^\top x\bigg\|^2 = \sum_{i = 1}^m x^\top V_iV_i^\top V_i V_i^\top x + 2\sum_{1 \leq i < j \leq m} (V_iV_i^\top x)^\top (V_j V_j^\top x).$$
    By utilizing the fact that the columns of $V_i$ are orthonormal, we can simplify $V_i^\top V_i = I$. Therefore, the first term in our sum simplifies as
    $$
    \sum_{i = 1}^m x^\top V_iV_i^\top V_i V_i^\top x = \sum_{i = 1}^m x^\top V_iV_i^\top x = \sum_{i = 1}^m \|V_i^\top x\|^2 = \lambda,
    $$
    as derived in the previous step.

    As the columns of $V_i$ form an orthonormal basis for $\mathcal{R}(A_i)$, we have that $V_iV_i^\top x = V_i(V_i^\top x) \in \mathcal{R}(A_i)$. Therefore, we can view $V_iV_i^\top x, V_jV_j^\top x$ as two elements of $\mathcal{R}(A_i), \mathcal{R}(A_j)$ and by definition of $\theta_{ij}$ and $\theta_{\text{H}}$ we have that
    $$|(V_iV_i^\top x)^\top (V_j V_j^\top x)| \leq \cos(\theta_{ij}) \|V_iV_i^\top x\| \|V_jV_j^\top x\| \leq \cos(\theta_\text{H}) \|V_iV_i^\top x\| \|V_jV_j^\top x\|.$$
    By the AM--GM inequality, we obtain  $\|V_iV_i^\top x\| \|V_jV_j^\top x\| \leq \frac{1}{2}(\|V_iV_i^\top x\|^2 +  \|V_jV_j^\top x\|^2)$, from where we obtain $$|(V_iV_i^\top x)^\top (V_j V_j^\top x)| \leq \frac{\cos(\theta_{\text{H}})}{2}( \|V_iV_i^\top x\|^2 + \|V_j V_j^\top x\|^2).$$
    By summing up over all pairs $i < j$, we obtain
    \begin{align*}
    &2\sum_{1 \leq i < j \leq m} (V_iV_i^\top x)^\top (V_j V_j^\top x) \leq \cos(\theta_{\text{H}}) (m - 1) \sum_{i = 1}^m \|V_iV_i^\top x\|^2 \\ &= \cos(\theta_{\text{H}}) (m - 1) \sum_{i = 1}^m x^\top V_i V_i^\top V_i V_i^\top x
    = \cos(\theta_{\text{H}}) (m - 1) \sum_{i = 1}^m x^\top V_i V_i^\top x  \\ & = \cos(\theta_{\text{H}}) (m - 1) \sum_{i = 1}^m \|V_i^\top x\|^2 = \cos(\theta_{\text{H}}) (m - 1) \lambda,
    \end{align*}
    where we again used the fact that $V_i$ has orthonormal columns. Analogously
    $$
    2\sum_{1 \leq i < j \leq m} (V_iV_i^\top x)^\top (V_j V_j^\top x) \geq -\cos(\theta_{\text{H}}) (m - 1) \lambda.
    $$
    \subsubsection*{Step~\ref*{step:four_final}} The previous step yields  $\lambda^2 \leq \lambda + \cos(\theta_{\text{H}}) (m - 1) \lambda$ and  $\lambda^2 \geq \lambda - \cos(\theta_{\text{H}}) (m - 1) \lambda$.
    As $S = VV^\top$ is clearly a positive semi-definite matrix, and $\lambda \neq 0$, we can safely divide by $\lambda$ to attain:
    $$
    1 - \cos(\theta_{\text{H}}) (m - 1)\leq \lambda \leq 1 + \cos(\theta_{\text{H}}) (m - 1).
    $$
    Therefore, the ratio of the maximum and the minimum non-zero eigenvalues of $S$ cannot exceed $$\frac{1 + (m - 1)\cos(\theta_{\text{H}}) }{1 - (m - 1)\cos(\theta_{\text{H}}) }.$$ Since $S=VV^\top$ is positive semi-definite, this completes the proof.
\end{proof}

\subsection{Proof of Theorem~\ref*{prop:a_transpose_a}}
\begin{proof}
    We perform the steps outlined below Theorem~\ref*{prop:a_transpose_a}.
    \subsubsection*{Step~\ref*{step:one}} Identify a sub-matrix $A_{\mathcal{S}}$ of $A$ as described. As $A_{\mathcal{S}}$ is a square matrix, it has the same singular values as $A_\mathcal{S}^\top$. Therefore, by the Courant-Fischer theorem for singular values, we have
    $$
    \sigma_{\min}(A_{\mathcal{S}}) = \sigma_{\min}(A_{\mathcal{S}}^\top) = \inf_{\|x\| = 1} \|A_{\mathcal{S}}^\top x\|
    $$
    and analogously
    $$
    \sigma_{\max}(A_{\mathcal{S}}) = \sigma_{\max}(A_{\mathcal{S}}^\top) = \sup_{\|x\| = 1} \|A_{\mathcal{S}}^\top x\|
    $$

    \subsubsection*{Step~\ref*{step:two}}

    Fix some $k \in \mathcal{S}$. Noting that $\|e_k\| = 1$, by substituting it for $x$ in $\sup_{\|x\| = 1}\|A_{\mathcal{S}}^\top x\|$, we obtain $\sigma_{\max}(A_{\mathcal{S}}) = \sup_{\|x\| = 1} \|A_{\mathcal{S}}^\top x\| \geq \|A_{\mathcal{S}}^\top e_k\| = \|a_k\|$. Therefore $\sigma_{\max}(A_{\mathcal{S}}) \geq \|a_k\|$. As $k \in \mathcal{S}$ is arbitrary, we conclude that $\sigma_{\max}(A_{\mathcal{S}}) \geq \max_{k \in \mathcal{S}} \|a_k\|$.

    \subsubsection*{Step~\ref*{step:three}}

    Fix any $\ell \in \mathcal{S}$. Notice that any vector of the form $x = ce_\ell + \sqrt{1 - c^2} e_k$ for $|c|\le 1$ and $\ell \neq k \in \mathcal{S}$ satisfies $\|x\| = 1$. Therefore $C := \{ce_\ell + \sqrt{1 - c^2} e_k : c \in [-1, 1] \setminus \{0\}, \ell \neq k \in \mathcal{S}\} \subseteq \{x \in \mathbb{R}^n : \|x\| = 1\}$. It immediately follows that $\sigma_{\min}(A_{\mathcal{S}})^2 \leq \inf_{x \in C} \|A_{\mathcal{S}}^\top x\|^2$. Pick the value of $k$ for which the angle between $a_k$ and $a_\ell$ is $\theta^{(\ell)}_{\min}$. Then, for any $c \in [-1, 1] \setminus \{0\}$ we have  $$\sigma_{\min}(A_{\mathcal{S}})^2 \leq \|ca_\ell + \sqrt{1 - c^2}a_k\|^2 = c^2 \|a_\ell\|^2 + (1 - c^2) \|a_k\|^2 + 2c\sqrt{1 - c^2}a_\ell^\top a_k.$$ As $[-1,1] \setminus \{0\} \ni c\to \frac{\sqrt{1 - c^2}}{c}\in \R$ is a surjective function, there exists a value of $c$ for which  $$\frac{\sqrt{1 - c^2}}{c} = -\mathrm{sgn}(a_\ell^\top a_k) \cos(\theta^{(\ell)}_{\text{min}}) \frac{\|a_\ell\|}{\|a_k\|}.$$ Then $(1 - c^2)\|a_k\|^2 = c^2\|a_\ell\|^2 \cos^2(\theta^{(\ell)}_{\text{min}})$ and $2c\sqrt{1-c^2} a_{\ell}^\top a_k = -2c^2 |a_\ell^\top a_k| \cos(\theta^{(\ell)}_{\text{min}}) \frac{\|a_\ell\|}{\|a_k\|} = -2c^2 \cos^2(\theta^{(\ell)}_{\text{min}}) \|a_\ell\|^2$. Hence, we obtain
    $$\sigma_{\min}^2(A_{\mathcal{S}}) \leq c^2\|a_\ell\|^2 + c^2 \|a_\ell\|^2 \cos^2(\theta^{(\ell)}_{\text{min}}) - 2c^2 \|a_\ell\|^2 \cos^2(\theta^{(\ell)}_{\text{min}}) = c^2 \|a_\ell\|^2 \sin^2(\theta^{(\ell)}_{\text{min}})$$

    Using that $c^2 \in [0, 1]$ by definition; that $\sigma_{\min}(A_{\mathcal{S}})$, $\|a_{\ell}\|$, and $ \sin(\theta^{(\ell)}_{\text{min}})$ are all non-negative; and that $\ell$ is arbitrary, we attain the desired bound of $\sigma_{\min}(A_{\mathcal{S}}) \leq \min_{\ell \in \mathcal{S}} \{ \|a_\ell\| \sin(\theta^{(\ell)}_{\min})\}$.

    \subsubsection*{Step~\ref*{step:four}}

    As $\sigma_{\min}(A_{\mathcal{S}}) > 0$ since $A_{\mathcal{S}}$ is non-singular, we can safely divide our lower bound for $\sigma_{\max}(A_{\mathcal{S}})$ and our upper bound for $\sigma_{\min}(A_{\mathcal{S}})$ to attain that $$\kappa(A_{\mathcal{S}}) = \frac{\sigma_{\max}(A_{\mathcal{S}})}{\sigma_{\min}(A_{\mathcal{S}})} \geq \frac{\max_{k \in \mathcal{S} }\|a_k\|}{\min_{\ell \in \mathcal{S}} \{ \|a_\ell\| \sin(\theta^{(\ell)}_{\min})\}}.$$
\end{proof}

\bibliographystyle{siamplain}
\bibliography{references_supplement}